\documentclass[runningheads]{llncs} 
\usepackage{thm-restate}

\makeatletter
\RequirePackage[bookmarks,unicode,colorlinks=true]{hyperref}%
   \def\@citecolor{blue}%
   \def\@urlcolor{blue}%
   \def\@linkcolor{blue}%

\def\orcidID#1{\smash{\href{http://orcid.org/#1}{\protect\raisebox{-1.25pt}{\protect\includegraphics{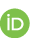}}}}}
\makeatother 

\usepackage{amsmath,amssymb}
\usepackage{stmaryrd}
\usepackage{todonotes}
\usepackage{cite}
\usepackage{tcolorbox}
\usepackage{xspace}
\usepackage{pifont}
\usepackage{mathpartir}
\usepackage[capitalise]{cleveref}
\usepackage{thm-restate}
\usepackage{url}

\usepackage[inline]{enumitem}
\usepackage{geometry}

\definecolor{mycolor}{rgb}{0.122, 0.435, 0.698}
\definecolor{light-green}{rgb}{0.0,0.5, 0.0} 
\definecolor{dark-green}{rgb}{0.0, 0.42, 0.24}

\newtcbox{\mybox}{on line,
  colframe=mycolor,colback=mycolor!10!white,
  boxrule=0.5pt,arc=4pt,boxsep=0pt,left=3pt,right=3pt,top=3pt,bottom=3pt}

\newcommand{\OK}{{\color{light-green}\text{\ding{51}}}\xspace}
\newcommand{\FAIL}{{\color{red}\text{ \ding{55}}}\xspace}
\newcommand{\Seq}{\mathcal{S}}

\newcommand{\true}{\ensuremath{\mathsf{true}}}

\newcommand{\mathsc}[1]{{\normalfont\textsc{#1}}}
\newcommand{\lres}{\Lambda}
\newcommand{\sort}[1]{\textsc{#1}}
\newcommand{\code}[1]{\texttt{#1}}

\newcommand{\eqdef}{\mathrel{\widehat{=}}}
\newcommand{\nat}{\mathbb{N}}

\def\ndres{\mathbin{\rlap{\raise.05ex\hbox{$-$}}{\lhd}}}

 \makeatletter
 \newcommand{\oless}{\mathbin{\mathpalette\make@circled <}}
 \newcommand{\make@circled}[2]{%
   \ooalign{$\m@th#1\smallbigcirc{#1}$\cr\hidewidth$\m@th#1#2$\hidewidth\cr}%
 }
 \newcommand{\smallbigcirc}[1]{%
   \vcenter{\hbox{\scalebox{0.77778}{$\m@th#1\bigcirc$}}}%
 }
 \makeatother

\newcommand{\loc}{x}

\newcommand{\val}{\kappa}

\newcommand{\Exp}{\sort{Exp}}
\newcommand{\BExp}{\sort{BExp}}
\newcommand{\Reg}{\sort{Reg}}
\newcommand{\Val}{\sort{Val}}
\newcommand{\Loc}{\sort{Loc}}
\newcommand{\xxTId}{\sort{Tid}}

\newcommand{\TId}{\xxTId}

\newcommand{\tid}{\ensuremath{\tau}}

\newcommand{\expeval}[2]{\llbracket #1\rrbracket_{#2}}

\newcommand{\skipst}{\kwc{skip}}
\newcommand{\kwc}[1]{\textnormal{\textbf{#1}}}

\newcommand{\load}[2]{#1 := \kwc{load}\,#2}
\newcommand{\store}[2]{\kwc{store}\,\,#1\,\, #2}

\newcommand{\new}[1]{#1}

\newcommand{\vwnew}{v_{\mathit{wNew}}}
 \newcommand{\vwold}{v_{\mathit{wOld}}}
 \newcommand{\vread}{v_{\mathit{read}}}

\newcommand{\imp}{\Rightarrow}

\newcommand{\sdef}{\mathrel{\widehat{=}}}
\newcommand{\lfun}{\ell} 
\newcommand{\Act}{\mathit{Act}}
\newcommand{\TS}{\mathbb{T}}
 
\newcommand{\sem}[1]{\llbracket #1 \rrbracket}

\newcommand{\linefill}{\cleaders\hbox{$\smash{\mkern-2mu\mathord-\mkern-2mu}$}\hfill\vphantom{\lower1pt\hbox{$\rightarrow$}}}  
  
\newcommand{\transi}[2]{\mathrel{\lower1pt\hbox{$\mathrel-_{\vphantom{#2}}\mkern-8mu\stackrel{#1}{\linefill_{\vphantom{#2}}}\mkern-11mu\rightarrow_{#2}$}}}
\newcommand{\trans}[1]{\transi{#1}{{}}}

\newcommand{\causal}{\twoheadrightarrow}

\newcommand{\rdflow}[1]{\mathbin{{\color{red} \stackrel{#1}{\causal}}}}

\newcommand{\conflict}{\#}
 
\newcommand{\cE}{{\cal E}}

\newcommand{\conf}{\mathit{Conf}}
\newcommand{\addc}{\mathit{rstr}}
 
\newcommand{\reg}{a}
\newcommand{\rega}{b}

\usepackage[inference]{semantic}

\usepackage{tikz}
\usetikzlibrary{arrows,snakes,backgrounds,automata,shapes,positioning,chains,calc}
\usetikzlibrary{arrows}

\tikzset{
    cf/.style={decorate, decoration={
                            zigzag,
                            segment length=4,
                            amplitude=.9
                          },
                >=stealth,thick,black!20!purple},
    cs/.style={dashed,->,>=stealth,thick,black!50!green} 
 }

\newcommand{\ffkw}{\mathsf{ff}} 
\newcommand{\prkw}{\mathsf{rd}} 
\newcommand{\inikw}{\mathsf{ini}} 
\newcommand{\testkw}{\mathsf{tst}}

\newcommand{\priorFnc}{\mathit{prFnc}}
\newcommand{\priorBar}{\mathit{prBar}}
\newcommand{\priorTest}{\mathit{prTst}}
\newcommand{\bexp}{\beta}

 \newcommand{\prm}[3]{\prkw_{#1}(#2,#3)}
  \newcommand{\prmts}[4]{\prkw_{#1}^{#2}(#3,#4)}
  
  \newcommand{\Ini}{\mathsf{Ini}}
 \newcommand{\ff}[3]{\mathsf{ff}_{#1}(#2,#3)}
 \newcommand{\ffts}[4]{\mathsf{ff}_{#1}^{#2}(#3,#4)}
 
 \newcommand{\barrx}[2]{\mathsf{bar}(#1,#2)} 
 \newcommand{\barrexp}[2]{\mathsf{bar}(#1,#2)} 
  
 \newcommand{\fence}{\mathsf{fnc}}
 \newcommand{\test}[2]{\mathsf{tst}_{#1}(#2)} 

 \newcommand{\lst}[2]{\mathit{lst}(#1,#2)}

 \renewcommand{\Pr}{\mathsf{Rd}} 
 \newcommand{\Ff}{\mathsf{Ff}}

\newcommand{\cat}{\mathbin{+\!\!\!+}}

\newcommand{\rv}{{\it rv}}
\newcommand{\expr}{\eta}
\newcommand{\assume}{\kwc{asm}}
\newcommand{\inarr}[1]{\begin{array}{@{}c@{}}#1\end{array}}
\newcommand{\inarrL}[1]{\begin{array}{@{}l@{}}#1\end{array}}

\title{Reasoning about Promises in Weak Memory Models with Event Structures (Extended Version)\thanks{Wehrheim and Bargmann are supported by DFG-WE2290/14-1. Dongol is supported by EPSRC grants EP/V038915/1, EP/R032556/1, EP/R025134/2, VeTSS and ARC Discovery Grant DP190102142.   }}
\author{Heike Wehrheim\orcidID{0000-0002-2385-7512}\inst{1} \and Lara Bargmann\inst{1} \and Brijesh Dongol\orcidID{0000-0003-0446-3507}\inst{2}}
\institute{University of Oldenburg, Oldenburg, Germany \and  University of Surrey, Guildford, UK}

\begin{document}
\maketitle

\begin{abstract}
  Modern processors such as ARMv8 and RISC-V allow executions in which
  independent instructions within a process may be reordered. To cope
  with such phenomena, so called {\em promising} semantics have been
  developed, which permit threads to read values that have not yet
  been written. Each promise is a speculative update that is later
  validated (fulfilled) by an actual write.  Promising semantics are
  operational, providing a pathway for developing proof calculi. In
  this paper, we develop an incorrectness-style logic, resulting in a
  framework for reasoning about state reachability. Like incorrectness
  logic, our assertions are
  \emph{underapproximating}, 
  since the set of all
  valid promises are not known at the start of execution. Our logic
  uses {\em event structures} as assertions  to compactly represent the
  ordering among events such as promised and fulfilled writes. We
  prove soundness and completeness of our proof calculus and
  demonstrate its applicability by proving reachability properties of
  standard weak memory litmus tests.
\end{abstract}


\section{Introduction} 


In recent years, numerous works have looked into semantics for
weak memory models for various hardware architectures or languages, e.g.~for
x86-TSO~\cite{DBLP:journals/cacm/SewellSONM10},
C11~\cite{DBLP:conf/popl/BattyOSSW11,DBLP:conf/oopsla/NienhuisMS16},
Power~\cite{DBLP:conf/pldi/SarkarSAMW11} or
ARM~\cite{DBLP:conf/popl/FlurGPSSMDS16}. Such semantics typically can
be classified as either being declarative (aka axiomatic) or
operational. Operational semantics furthermore can be divided into
those following a microarchitectural style (providing formalizations
of the actual hardware architecture) and those trying to abstract from
architectures. Most notably, {\em view-based}
semantics~\cite{DBLP:conf/ecoop/KaiserDDLV17,DBLP:conf/ppopp/DohertyDWD19,PultePKLH19}
avoid modelling specific hardware components and instead define the
semantics in terms of {\em views} of thread on the shared state. {\em
  Promises}~\cite{KangHLVD17,DBLP:conf/pldi/LeeCPCHLV20} are employed
in operational semantics as a way of capturing out-of-order writes
while still executing operations in thread order.  A promise (w.r.t.~a
value $\val$ and a shared location $x$) of a thread $\tau$ states that $\tau$
will eventually write value $\val$ onto location $x$. All promised writes
then need to be {\em fulfilled} (i.e., justified) in the future of a program
run, but other threads can read from promises before they are
fulfilled.


Our
interest here is the development and use of
Hoare-style~\cite{DBLP:journals/cacm/Hoare69} structural proof calculi  
(and their extensions to concurrency by Owicki and
Gries~\cite{DBLP:journals/acta/OwickiG76}) for weak memory models. 
Owicki-Gries-like proof calculi have been proposed by a number of
researchers~\cite{DBLP:conf/icalp/LahavV15,DalvandiDDW19,DBLP:journals/corr/abs-2004-02983,DBLP:conf/fm/WrightBD21},
and have also recently been given for non-volatile
memory~\cite{DBLP:journals/pacmpl/RaadLV20,esop2022}. Svendsen et
al.~\cite{DBLP:conf/esop/SvendsenPDLV18} have developed a separation
logic for promises for the C11 memory model. 
Wright et al.~\cite{DBLP:conf/fm/WrightBD21} have developed an Owicki-Gries proof system for out-of-order writes (as allowed by promises), but rely on pre-processing via the denotational MRD framework~\cite{DBLP:conf/esop/PaviottiCPWOB20}.

All of these proposals follow Hoare's principle of providing {\em safety}
proofs. In particular, a Hoare triple
${\color{light-green} \big\{p \big\}} S {\color{light-green} \big\{
  q\big\} } $ describes the fact that an execution of program $S$
starting in a state satisfying $p$ is either non-terminating, or
terminates in a state satisfying $q$ (over-approximating the final
states). However, for weak memory models, we often want to prove
{\em reachability}, i.e.~under-approximate the set of final states, like in
the recent proposal of O'Hearn's incorrectness
logic~\cite{OHearn20}. Here, a triple
${\color{mycolor} \big[ p \big]} S {\color{mycolor}\big[ q \big] }$
describes the possibility of program $S$ reaching all states
satisfying $q$ when started in a state satisfying $p$.  A verification
technique supporting these {\em reachability} triples enables one to
reason about 
executions that deviate from the expected sequentially consistent
behaviour of concurrent programs.

\smallskip\noindent \textit{Contributions.}
In this paper, we present a reachability proof calculus for concurrent
programs where the semantics of the weak memory model is based on
promises. The specific challenges therein lay in (i) capturing the
meaning of promises as writes which will only happen in the future but
can nevertheless already be read from, and (ii) appropriately
describing the required ordering (and concurrency) between  
promises and fulfills as fixed by the concurrent program under
consideration. We address these challenges via the following contributions. 
\begin{enumerate*}[label=(\arabic*)]
\item We develop a program logic based on assertions which are (flow) {\em event
  structures}~\cite{Winskel88,BoudolC94,GlabbeekG04}, employing
parallel composition of event structure and synchronization as a means
of determining whether all promises read from have eventually been
fulfilled.
\item We extend the theory of flow event structures with the notion of
  a flow label to capture the behaviours observed in weak memory
  models.
  \item We develop the first {\em compositional} proof rule for a
  concurrent reachability (incorrectness) logic.
\item We prove soundness and completeness of this novel event-structure based proof calculus.
\item Finally, we demonstrate its applicability  on a number of litmus tests. 
\end{enumerate*}



\smallskip\noindent\textit{Overview.} In \cref{sec:motivating-example}, we provide a
concrete overview via a motivating example and in
\cref{sec:weak-memory-semant}, we present the memory model that we
use. Our model is a simplified (strengthened) version of the
ARMv8/RISC-V semantics of Pulte et al~\cite{PultePKLH19}. In
\cref{sec:event-structures}, we present an extended theory for event
structures (specifically an extension of flow event structures) that has been designed to enable reasoning about relaxed memory models. We describe our reasoning methodology and provide examples verifying common litmus tests in \cref{sec:reasoning}.

\section{Motivating Examples}
\label{sec:motivating-example}

  Consider the program in \cref{fig:load-buffering}, which describes
  the load buffering litmus test. Thread 1 (similarly thread 2) loads
  the value of $y$ (sim. $x$) into register $\reg$ (sim. $\rega$),
  then updates $x$ (sim. $y$) to $1$. Since there are no dependencies
  between lines 1 and 2, and similarly between lines~3 and 4,
  architectures such as ARMv8 and RISC-V allow the stores in both
  threads to be reordered with the loads. Thus the program allows the
  final outcome {\color{light-green}$\reg = 1 \wedge \rega =
  1$}.  
  
This phenomenon is captured by promising semantics by allowing each thread to ``promise'' their respective stores, then later fulfilling them. In the meantime, other threads may read from promised writes.  
Our assertions within a thread reflect this semantics via assertions {\color{mycolor} $\cE$} which are {\em flow event structures}~\cite{Winskel88}. The events and their partial order reflect program executions, and in particular describe the various {\em views} which threads have on shared state. 

\begin{figure}[t]
    \centering \footnotesize
    \begin{minipage}[b]{0.49\columnwidth}
    \scalebox{0.85}{
    \begin{tabular}[b]{@{}c@{}} 
   
  $\begin{array}{l@{}||@{}l}
    \begin{array}[t]{l}
     \textbf{Thread } 1 
     \\
     {\color{mycolor} \big[\ \inikw\ \big] } \\ 
     1: \load{\reg}{y} ;
     \\
    {\color{mycolor} \big[\ \inikw \causal  \prm{2}{y}{1} \causal \barrx{\reg}{y}\ \big] } \\
     2: \store{x}{1} ;
     \\ 
   {\color{mycolor} \left[  \raisebox{0.8cm}{\begin{tikzpicture}[anchor=north,baseline,node distance=.4cm]
     \node (0) {$\inikw$ };
     \node [right = of 0] (1) {$\prm{2}{y}{1} \causal \barrx{\reg}{y}$};
     \node [below = of  1] (2) {$\ff{1}{x}{1}$};
     \path [->>, thick] (0) edge (1); 
     \path [->>, thick] (0) edge (2.west); 
  \end{tikzpicture}} \right]  }
     \end{array}
    & 
    \begin{array}[t]{l}
     \textbf{Thread } 2 
     \\
    {\color{mycolor} \big[ \dots \big] } \\
     3: \load{\rega}{x} ;
     \\
     {\color{mycolor} \big[ 
     \dots \big] }
    \\
    4: \store{y}{1} 
     \\
   {\color{mycolor} \big[ \dots \big]  } 
\end{array}
   \end{array}$ \\
   \qquad \qquad ${\color{mycolor} \left[  \begin{tikzpicture}[anchor=south,baseline,node distance=.6cm]
     \node (0) {$\inikw$ };
     \node[right =of 0] (1) {$\ff{1}{x}{1} \causal \barrx{\rega}{x}$};
     \node (2) [below of=1] {$\ff{2}{y}{1} \causal \barrx{\reg}{y}$};
     \path [->>,thick] (0) edge (1); 
     \path [->>,thick] (0) edge (2.west); 
   \end{tikzpicture}  \right] }$  \\
   \qquad \qquad \qquad \qquad ${\color{light-green} \big(a=1 \wedge b=1\big) \OK } $
    \end{tabular}}
    \caption{Reachability for load buffering}
    \label{fig:load-buffering}
    \end{minipage}    
    \begin{minipage}[b]{0.49\columnwidth}
    \scalebox{0.85}{
    \begin{tabular}[b]{@{}c@{}} 
   
  $\begin{array}{l@{}||@{}l}
    \begin{array}[t]{l}
     \textbf{Thread } 1 
     \\
     {\color{mycolor} \big[\ \inikw\ \big] } \\ 
     1: \load{\reg}{y} ;
     \\
    {\color{mycolor} \left[\ \inikw \causal  \prm{2}{y}{1} \causal \barrx{\reg}{y}\ \right] } \\
     2: \kwc{dmb} ;
     \\ 
   {\color{mycolor} \left[  \begin{tikzpicture}[anchor=south,baseline,node distance=.6cm]
     \node (1) {$\inikw \causal \prm{2}{y}{1}$} ;  
     \node (12) [right =0.4cm of 1] {$\barrx{\reg}{y}$};
     \node (2) [below of=1] {$\fence_1$};
     \path[->>,thick] (12) edge (2);
     \path[->>,thick] (1) edge (12);   
   \end{tikzpicture}  \right]  }  
   \\
   3: \store{x}{1} ;
     \\
   {\color{mycolor} \left[ \!\! \begin{tikzpicture}[anchor=south,baseline,node distance=.6cm]
     \node (1) {$\inikw \causal \prm{2}{y}{1}$} ;  
     \node (12) [right =0.4cm of 1] {$\barrx{\reg}{y}$};
     \node (2) [below of=1] {$\fence_1$};
     \node (22) [right =0.4cm of 2] {$\ff{1}{x}{1}$};
     \path[->>,thick] (12) edge (2.north east);
     \path[->>,thick] (1) edge (12);   
     \path[->>,thick] (2) edge (22);   
   \end{tikzpicture}  \right]  }
     \end{array}
    & 
    \begin{array}[t]{l}
     \textbf{Thread } 2 
     \\
    {\color{mycolor} \big[ \dots \big] } \\
     4: \load{\rega}{x} ;
     \\
    {\color{mycolor} \big[ \dots \big] } \\
     5: \kwc{dmb} ; \\
     {\color{mycolor} \big[ \dots \big] }
    \\
    6: \store{y}{1} 
     \\
      
   {\color{mycolor} \big[ 
    \dots
     \big]  } 
     \end{array}
   \end{array}$ \\
   \qquad\qquad\qquad \quad \ ${\color{red} \big[\  \cE \ \big] }$ 
     \\
   \qquad\qquad\qquad \qquad ${\color{red}  \big( a=1 \wedge b=1\big) \FAIL    } $
    \end{tabular} }
    \caption{Load buffering with barriers}
    \label{fig:load-buffering-fnc}
  \end{minipage}
  \end{figure}

The proof outlines (i.e., program texts with assertions)  of individual threads may first of all contain read events for arbitrary promises, i.e.~describe the reading of arbitrary values.  
In Thread 1 of \cref{fig:load-buffering}, the pre-assertion of the load only
contains an event for initial writes (labelled $\inikw$), yet the load may read the value~$1$ for
$y$ from a promised write, described by the event labelled $\prm{2}{y}{1}$ in
the post assertion. 
The  semantics
generates dependencies if the same register is used (perhaps
indirectly) by a read and a later write. This is captured in our
assertions using the event labelled $\barrx{\reg}{y}$, causally ordered after $\prm{2}{y}{1}$, which states that the
view of register $\reg$ is at least that of the read of 
$y$. 
Execution of line 2 then adds a fulfill event
with label $\ff{1}{x}{1}$ to the assertion, which is not ordered with any other event except $\inikw$. Symmetric
assertions can be generated for Thread 2. To obtain an
assertion describing the combined execution, we compose the final
event structures of both threads to obtain a ``postcondition'' of the
program. For this, we use parallel composition of event structures, synchronising read with their corresponding fulfill events. In \cref{fig:load-buffering}, 
both reads are valid since the promises that these reads rely on
can be fulfilled in the composition without creating cyclic dependencies. 

\cref{fig:load-buffering-fnc} presents a variation of the program in \cref{fig:load-buffering}, which includes additional barriers $\kwc{dmb}$ (fences) between the load and store in each thread, preventing their reordering. Again we build a proof outline for an execution in which Thread 1 loads $1$ into $\reg$, obtaining the assertions shown. Note that here the  event structure contains an additional fence event, $\fence$, that is ordered after $\barrx{\reg}{y}$ and before   $\ff{1}{x}{1}$.  Similarly, for Thread~2 loading $1$ into $\rega$, we would obtain a symmetric set of assertions. Here, the parallel composition of local assertions is however not {\em interference free} (see below): the promises that threads 1 and 2 have read from cannot be fulfilled in this concurrent program. More detailedly,  let $\cE_1$ and $\cE_2$ below be the (final) event structures of threads 1 and 2, respectively, where $\causal$ arrow denotes  ordering and we now give event names together with labels. 

\begin{minipage}[t]{1.0\linewidth}
$\cE_1$:
\begin{tikzpicture}[anchor=text,baseline,node distance=.6cm]
  \node (1) {$e_\inikw: \inikw \causal e_1: \prm{2}{y}{1}
    \causal e_2: \barrx{\reg}{y} \causal e_3: \fence_1 \causal e_4: \ff{1}{x}{1}$} ;  
   \end{tikzpicture}

   $\cE_2$:
\begin{tikzpicture}[anchor=text,baseline,node distance=.6cm]
\node (1) {$f_\inikw: \inikw \causal f_1: \prm{1}{x}{1}
    \causal f_2: \barrx{\rega}{x} \causal f_3: \fence_2 \causal f_4: \ff{2}{y}{1}$} ;  
   \end{tikzpicture}
 \end{minipage}

\noindent To reason about the set of reachable final states of the concurrent program, we again construct the parallel composition of $\cE_1$ and $\cE_2$ (denoted $\cE_1 \| \cE_2$):

\medskip
\noindent 
\scalebox{0.85}{
   \begin{tikzpicture}[anchor=south,baseline,node distance=.4cm]
  \node (00) {$(e_\inikw, f_\inikw): \inikw$} ;  
  \node [above right= 0.1cm and 0.5cm of 00] (14) {$(e_1, f_4): \ff{2}{y}{1}$} ;  
  \node [below right= 0.1cm and 0.5cm of 00] (41) {$(e_4, f_1): \ff{1}{x}{1}$} ;  

  \node [above =of 14] (1*) {$(e_1, *): \prm{2}{y}{1}$} ;  
  \node [right=of 1*] (2*) {$(e_2, *): \barrx{\reg}{y}$} ;  
  \node [right=of 2*] (3*) {$(e_3, *): \fence_1$} ;  
  \node [right=of 3*] (4*) {$(e_4, *): \ff{1}{x}{1}$} ;  

  \node [below =of 41] (*1) {$(*, f_1): \prm{1}{x}{1}$} ;  
  \node [right=of *1] (*2) {$(*, f_2): \barrx{\rega}{x}$} ;  
  \node [right=of *2] (*3) {$(*, f_3): \fence_2$} ;  
  \node [right=of *3] (*4) {$(*, f_4): \ff{2}{y}{1}$} ;

  \draw[->>,thick] (00) -- (1*.west); 
  \draw[->>,thick] (1*) -- (2*); 
  \draw[->>,thick] (2*) -- (3*);
  \draw[->>,thick] (3*) -- (4*); 

  \draw[->>,thick] (00) -- (*1.west); 
  \draw[->>,thick] (*1) -- (*2); 
  \draw[->>,thick] (*2) -- (*3);
  \draw[->>,thick] (*3) -- (*4); 

  \draw[->>,thick] (00) -- (14); 
  \draw[->>,thick] (00) -- (41); 

  \draw[->>,thick] (41) -- (*2); 
  \draw[->>,thick] (14) -- (2*);

  \draw[->>,thick] (*3) -- (14); 
  \draw[->>,thick] (3*) -- (41);

     \begin{scope}[draw=orange,
              every node/.style={fill=white,circle},
              every edge/.style={decorate,decoration={zigzag,segment length=2mm, amplitude=1mm},thick},
              every path/.style={decorate,decoration={zigzag,segment length=2mm, amplitude=1mm},thick}]
    \draw (1*) -- (14);
    \draw (41) -- (*1);
    \draw (14.east) -- (*4);
    \draw (41.east) -- (4*);
    \end{scope}

\end{tikzpicture}}

\bigskip This composition of event structure is built similar to
\cite{GlabbeekG04}, allowing events of the parallel composition to be
lifted from the sub-components. These are events of the form
$(e_i, *)$ and $(*, f_i)$. The parallel composition   also
contains \emph{synchronised} read/fulfill events, e.g., $(e_1, f_4)$ depicts a
  read synchronised with the fulfill (write)  $\ff{1}{y}{1}$. We inherit
order in the composition from the constituent event
structures. Moreover, to prevent the same event occurring more than
once in an ``execution'' of $\cE_1 \| \cE_2$, we use the {\em conflict} 
relation (zigzagged line). Thus, the synchronised event
$(e_1, f_4)$ conflicts with both $(e_1, *)$ and $(*, f_4)$.

The final step in proving is the generation of a
valid interference free {\em configuration} of the parallel composition, which is a subset of the event structure
satisfying certain conditions, including acyclicity of $\causal$, 
absence of conflicts and absence of unsynchronised reads (ensuring the fulfillment of all promises read from). It turns out that for the event structure above, it is
impossible to generate such a configuration. The event $(e_1, *)$
cannot be included since it is an unsynchronised read. Therefore,
$(e_1, f_4)$ must be included. However, by the definition of a
configuration, this also means that the downclosure of $(e_1, f_4)$
must be included, which results in a cycle:
$(e_1, f_4) \causal (e_2, *) \causal (e_3, *) \causal (e_4, f_1) \causal (*, f_2) \causal (*, f_3) \causal (e_1,
f_4)$. Since $\cE_1 \| \cE_2$ has no interference free configurations, the proof outline is not valid and in fact, a final state with 
$a=1 \wedge b=1$ is unreachable here. 






\section{A Weak Memory Semantics with Promises }
\label{sec:weak-memory-semant}

We develop a promising semantics inspired by the recent view-based
operational semantics 
by Pulte et
al.~\cite{PultePKLH19}. We have reduced architecture-specific details,
allowing us to focus on the interaction between promises and thread
views. Our notion of a {\em promise} coincides with earlier works
\cite{PultePKLH19,KangHLVD17,DBLP:conf/pldi/LeeCPCHLV20}. Threads can
promise to write certain values on shared locations and other threads
can read from this promise even before the actual write has
occurred. All promises however need to be fulfilled at the end of the
program execution.


\smallskip \noindent {\em Syntax.} Let $x,y \in \Loc$ be the set of {\em shared locations}, $\val \in \Val$ the set of {\em values}, $\tid \in \TId$ the set of {\em thread identifiers} and $\reg, \rega \in \Reg$ {\em local registers}. Our sequential language encompasses the following constructs:
\begin{eqnarray*}
  {\it rv} & ::= & \val \mid \reg  \qquad\qquad 
  st  ::=  \skipst \mid \load{\reg}{x} \mid \store{x}{{\it rv}}  \mid \reg := \expr \mid  \kwc{dmb} \\
  S & ::= & st \mid S; S \mid \assume\, \bexp \mid S + S \mid S^*
\end{eqnarray*}
where $\expr \in \Exp$ is an arithmetic and $\bexp \in \BExp$ is a
boolean expressions, both over (local) registers only. We assume
$S^* = \exists n \in \nat.\ S^n$, where $S^0 \eqdef \skipst$ and
$S^n \eqdef S ; S^{n-1}$.  We use abbreviations:
$\kwc{while} \, \bexp\, \kwc{do} \, S = (\assume\,\bexp; S)^*;
\assume\,\neg \bexp$ and
$\kwc{if} \, \bexp\, \kwc{then} \, S_1 \, \kwc{else} \, S_2 =
(\assume\,\bexp ; S_1) + (\assume\,\neg \bexp ; S_2)$, \new{where
  $\assume\,\bexp\,$ is a command that tests whether $\bexp$ holds.}

\smallskip \noindent {\em Timestamped state.} We let ${\it TState}$ be the set of all {\em timestamped states} and $\mathit{Memory}$ the set of all memory states, both of which we make more precise below.
A {\em thread} $T \in {\it Thread}$ is an element of $S \times {\it TState}$, a {\em concurrent program} is a mapping $\vec{T} \in {\it TPool} \ \eqdef\  \TId \rightarrow {\it Thread}$ and a {\em concurrent program state} is a pair $\langle \vec{T}, M \rangle \in \mathit{TPool} \times \mathit{Memory}$. 
We let $R(\tid)$, $\tid \in \TId$, be the set of registers occurring in the program of $\vec{T}(\tid)$. We assume $R(\tid) \cap R(\tid') = \emptyset$ whenever $\tid \neq \tid'$. 

Threads will make {\em promises} for writes at particular timestamps. Timestamps $t \in \TS$ are natural numbers.  
We define $t \sqcup t'\ \eqdef\ max(t,t')$ and generalise this to sets of timestamps using $\bigsqcup_{t\in T} t$, where $\bigsqcup_{t\in \emptyset} t = 0$. 
A {\em memory} is a sequence of write messages of type ${\sf Wr} \ \eqdef\ (\Loc \times \Val \times \TId) \cup \{\inikw\}$, where $\inikw$ is a special write message denoting initialisation. The position of a write in the sequence fixes its timestamp. 
We assume all variables are initialised with value $0$.

 


We denote a write $w \ \eqdef\ (x,\val,\tau)$ using $\langle x:= \val\rangle_\tid$ and let $w.loc=x, w.val = \val$ and $w.tid= \tau$. For a memory $M$ and thread $\tau \in \TId$, we let $M_\tau \subseteq {\TS}$ 
be the set of timestamps of entries of $\tau$ in $M$, i.e.~$\{ t \in \TS \mid M(t).tid = \tau\}$. $M_\tau$ is used to determine the promise set of each $\tid$. We write $tids(M)$ to denote the set of threads with entries in $M$. New messages $w$ are appended at the end of the memory, which we write as $M \cat w$. 

A thread state $ts \in {\it TState}$ consists of the following components: a set of (non-fulfilled) {\em promises} $prom \in 2^\TS$, a {\em coherence} view of each location, $coh: \Loc \rightarrow \TS$, the value and view of each register, \new{${\it regs}: \Reg \rightarrow \Val \times \TS$}, a read view \new{$\vread: \TS$}, two write views \new{$\vwold, \vwnew : \TS$} and a condition view \new{$v_C : \TS$}. We write ${\it regs}(\reg)$ as $\val@v$ and also let $v_\reg$ be this view  $v$ of register $\reg$. 
Finally, the evaluation of an expression $\expr$ with respect to a register assignment ${\it regs}$, \new{$\expeval{\expr}{\it regs}\in \Val \times \TS$}, is defined as follows:
\begin{align*}
    \expeval{\val}{\it regs} & \eqdef  \val@0 \text{ for } \val \in \Val, \qquad
    \expeval{\reg}{\it regs}  \eqdef  {\it regs}(\reg) \text{ for } \reg \in \Reg,  \\
    \expeval{\expr_1 \mathop{op} \expr_2 }{\it regs} & \eqdef 
    \begin{array}[t]{@{}l}
    (\val_1 \sem{op} \val_2)@(v_1 \sqcup v_2) \text{ with } 
    \expeval{\expr_1}{\it regs} = \val_1@v_1, \expeval{\expr_2}{\it regs} = \val_2@v_2 
    \end{array}
\end{align*} 
Note that this evaluation is with respect to the register function ${\it regs}$ and this calculates both the value of the expression and the maximal view of the registers within the expression. 

To define the initial state of a program, we let
\begin{align*}
   M_\inikw & \eqdef \langle \inikw \rangle & 
   ts_\inikw  & \eqdef  \left[
                \begin{array}[c]{@{}l@{}}
                  prom = \{\}, v_{read} = v_{wNew} = v_{wOld} = v_C = 0, \\
                  coh = (\lambda x.\ 0), {\it regs} = (\lambda \reg.\ 0 @ 0)
                \end{array}
\right]
\end{align*}
where $ts_\inikw$ is a record initialising the promises to the empty set, each view to $0$, the coherence function to a map from locations to timestamp $0$, and the register function to a map from registers to value $0$ with timestamp $0$.
We say that a program $\vec{T}$ is locally in its initial state iff for each thread $\tid$, we have $\pi_2(\vec{T}(\tid)) = ts_\inikw$, where $\pi_i$ projects the $i$th component of a tuple. Given that $\vec{T}$ is in its initial state, the initial concurrent program state is given by  
$\langle\vec{T}, M_\inikw\rangle$.

\begin{figure}[t]
\footnotesize

  \[ \inferrule[Promise]{
      \inarr{
        w.tid = \tid \quad w.loc = \loc \quad w.val =\val \\
        t = |M| + 1 \\ ts' = ts[prom \mapsto ts.prom \cup \{t\}] }
      }
    {
    \begin{array}{l}
\big\langle \langle S,ts \rangle, M \big \rangle \trans {prm(x,\val)}_{\tid} \big\langle \langle S,ts' \rangle, M \cat w \big \rangle
    \end{array}
     } 
 \quad 
 \inferrule[Fence]{
   \inarr{
     v = ts.v_{read} \sqcup ts.v_{wOld} \\
     ts' = ts
     \ensuremath{\left[\begin{array}[c]{l}
                         v_{read} \mapsto v \\ 
                         v_{\mathit{wNew}} \mapsto    v
                       \end{array}\right]}
                 }}
{ \big \langle \langle \kwc{dmb},ts\rangle, M \big \rangle \trans {\fence}_{\tid} \big \langle \langle \skipst,ts'\rangle, M \big \rangle  }
\] \vspace{-20pt}

\[
  \inferrule[Read]{
    \inarr{
      M(t).loc = x \qquad  M(t).val = \val 
      \\ 
      \forall t'.\ t < t' \le (ts.v_{read} \sqcup ts.coh(x)) \imp\\
      \qquad M(t').loc \neq \loc 
      \\
      v_{post} = ts.v_{read} \sqcup t \\ 
      ts'=ts \left[
        \inarrL{regs(\reg) \mapsto \val @ v_{post}, \\
          coh(\loc) \mapsto ts.coh(x) \sqcup v_{post},  \\ 
          v_{read} \mapsto v_{post}}
      \right]
    }}
  { 
    \inarr{
      \big \langle \langle \load{\reg}{\loc}, ts \rangle, M \big \rangle \trans {rd(\loc,\val)}_{\tid} \qquad \\
      \hfill \big \langle  \langle \skipst,ts'\rangle, M \big \rangle
    }
  } 
  \quad
  \inferrule[Fulfill]{
     \inarr{
      t  \in ts.prom \quad \expeval{\rv}{ts.regs} = \val @ v_\rv \\
      M(t)=\langle x :=\val\rangle_{\tid} \\ 
      ts.\vwnew \sqcup ts.v_C \sqcup ts.coh(x) \sqcup v_\rv  < t 
      \\
      ts'=ts
      \left[\inarrL{
          prom \mapsto ts.prom \setminus \{t\}, \\ 
          coh(x)\mapsto  t , \\
          \vwold \mapsto  \vwold   \sqcup t}
        \right]}}
  {
    \inarr{
      \big \langle \langle \store{x}{\rv}, ts\rangle ,M \big\rangle  \trans {{\it ff}(x,\val)}_{\tid} \qquad \\
      \hfill \big \langle \langle \skipst,ts'\rangle, M \big \rangle 
    }
  }
\]
\[
\hfill\inferrule[Register]{
       regs(\reg) = \val_\reg @ u \qquad \expeval{\expr}{ts.regs} = \val @ v \qquad
ts'=ts
     \ensuremath{\left[\begin{array}[c]{l}
       regs(\reg) \mapsto \val @ (u  \sqcup v)\\
     \end{array}\right]} 
}
{
  \begin{array}[t]{@{}l@{}}
    \big \langle \langle \reg := \expr, ts \rangle, M \big \rangle \trans {\lst{\reg}{\expr} }_\tau   \big \langle  \langle \skipst,ts'\rangle, M \big \rangle
  \end{array}
} \hfill
\]
\[
  \inferrule[Assume]{
    \inarr{
      \expeval{\bexp}{ts.regs} = \true @ v
      \\
      ts' = ts[v_C \mapsto ts.v_C \sqcup v]}}
  {\big \langle \langle \assume\,  \bexp, ts \rangle, M \big \rangle \trans {asm(\bexp)}_\tau \big \langle \langle \skipst, ts' \rangle, M \big \rangle }  
\quad 
\inferrule[Program Step]{
  \inarr{
    \langle \vec{T}(\tid),M \rangle \trans {op}_\tid \langle T', M \rangle \\
  \langle T', M \rangle \text{ certifiable }}}{ \langle \vec{T}, M \rangle \trans {op}_\tid \langle \vec{T}[ \tid  \mapsto T'], M' \rangle }
\]
  \caption{Operational semantics (Atomic statement rules)}

\label{fig:opsem} 
\end{figure}

The rules of the operational semantics (except for standard rules for program constructs) are given in \cref{fig:opsem}. The two key rules are the {\sc Read} and {\sc Fulfill} rule. {\sc Read} identifies a timestamp $t$ to read a value for $x$ from such that in between $t$ and the maximum of read view and coherence of $x$, there are no further promises to $x$ in memory $M$. It updates read view, coherence of $x$ and the view of the register involved in the load as to ensure preservation of dependencies. {\sc Fulfill} fulfills an already made promise (to write $\val$ to $x$) of a thread at timestamp $t$, and to this end has to ensure that views $\vwnew, v_C, coh(x)$ as well as that of the value/register are less than $t$. It removes $t$ from the thread's promise set and updates $coh(x)$ and $\vwold$ (as to ensure dependencies with fences).  
Rule {\sc Promise} simply adds an arbitrary new promise at the end of memory. {\sc Fence} ensures views $\vread$ and $\vwnew$ are updated. This rule for instance guarantees that store operations separated by barriers $\kwc{dmb}$ can only be fulfilled in that order, i.e.~the write of the first store cannot be promised to happen later than the write of the second store (more precisely, such promises cannot be fulfilled).  

Finally, we say that $\langle T, M \rangle$ is {\em certifiable} (used in {\sc Program Step}) if there is  
some $T', M'$ such that
$\langle T, M \rangle \trans {}^*_\tau \langle T', M' \rangle$ and
$T'.prom = \emptyset$. Certifiability ensures that a concurrent program can only make steps when all promises can eventually be fulfilled. Like~\cite{PultePKLH19}, in our semantics, all promise steps can be done at the beginning without losing any of the reachable states (see \cref{lemma:promises-first}).

\newcommand{\lab}{\alpha}

\section{Event Structures}
\label{sec:event-structures}

Event structures~\cite{Winskel88,BoudolC87,BoudolC94,GlabbeekG04} are
models of concurrent systems which compactly represent (concurrent)
executions. 
Here, we use
flow event structures because of their ease in defining a
compositional parallel composition~\cite{GlabbeekG04}. 

{\em Notation.} 
Event structures consist of sets of {\em events} $d,e,f \in E$.
Events will be labelled with {\em actions} which are here specific to our usage and give us information about program executions: 
\begin{small}
\begin{align*}
   \Act^x & \eqdef \bigcup_{\tid \in \TId, \val \in \Val} \{ \prm{\tau}{x}{\val},\ff{\tau}{x}{\val}\} \cup \{\inikw\} &
   \Act^{\fence} & \eqdef \bigcup_{\tid \in \TId} \{ \fence_\tau \} \\
   \Act^a & \eqdef \bigcup_{x \in \Loc, \expr \in \Exp} \{ \barrx{a}{x}, \barrexp{a}{\expr} \} & 
   \Act^\testkw & \eqdef \bigcup_{\tid \in \TId , \bexp \in \BExp} \{ \test{\tid}{\bexp} \} 
\end{align*}
\end{small}

\noindent Actions on a location $x$ can be {\em read} actions $\prm{\tid}{\cdot}{\cdot}$, {\em fulfill} actions $\ff{\tid}{\cdot}{\cdot}$ or the initialization $\inikw$. Note that the thread identifier $\tid$ in read actions is the id of the thread having made the promise and in fulfill actions it is the thread executing the fulfill (and having made the corresponding promise). We let $\Act^\prkw$ denote all read and $\Act^\ffkw$ all fulfill actions. To record loading into register $a$, we use so called {\em bar} actions $\barrx{a}{\cdot}$.  The action $\fence$  occurs when a \kwc{dmb} statement is executed
and  $\test{\cdot}{\cdot}$ describes the execution of some $\assume$ statement. 

We often lift notations to sets of locations $L \subseteq \Loc$ or sets of registers $R \subseteq \Reg$. For example, $\Act^L = \bigcup_{x \in L} \Act^x$. The overall set of actions is $\Act = \Act^\Loc \cup \Act^\Reg \cup \Act^\fence \cup \Act^\testkw$. 


\begin{definition}
 A  {\em location-coloured flow  event structure} (short: event structure) $\cE=(E,\causal,\conflict,\lres,\lfun)$ labelled over a set of actions $Act$ consists of 
 \begin{itemize*}[label={}]
     \item a finite set of events $E$,
     \item an irreflexive {\em flow relation}  $\mathord{\causal} \subseteq E \times E$, 
     \item a {\em location restriction} function $\lres : E \times E \to 2^{\mathsc{Loc}}$,  
     \item a symmetric  {\em conflict relation} $\# \subseteq E \times E$, and 
     \item a {\em labelling function}  $\lfun: E \rightarrow Act$.   
 \end{itemize*}
 \end{definition} 

  \noindent For $L \subseteq \Loc$, we write $e \rdflow{L} f$ to denote
 $e \causal f$ and $\lres(e, f) = L$. The location restrictions are employed to reflect the application condition of rule {\sc Read} within the event structure: it tells us that there is no write to $x \in L$ in between $e$ and $f$, where $e$ and $f$ will eventually be mapped to timestamps in memory. 
 
 We let $\Ini$ be the event
 structure $(\{e_{\inikw}\},\emptyset,\emptyset,\emptyset,\lfun)$ with
 $\lfun(e_\inikw) = \inikw$.
 Given an event structure
 $\cE= (E,\causal,\conflict,\lres,\lfun)$, we -- similarly to actions
 -- define its set of events labelled with specific actions as
 $\Pr(\cE)$, $\Pr_\tau(\cE)$, $\Pr_\tau^x(\cE)$, $\Ff(\cE)$,
 $\Ff_\tau(\cE)$ and $\Ff_\tau^x(\cE)$
 via the labelling function $\lfun$.  For an event $e$ labelled with
 an action in $\Act^x \backslash \{\inikw\}$, we let $e.loc =
 x$. We slightly abuse notation so that $e_\inikw.loc =
 x$ for all $x$. We furthermore define
 $last_{\lab}(\cE)$, $\lab \in Act$, to be the last event in flow order
 labelled $\lab$, i.e., $last_\lab(\cE) = e$ if $\lfun(e) = \lab$ and  
 for all  $e'$ such that $e' \neq e$ and $\lfun(e')=\lab$, we have  $e' \causal^+ e$. Moreover, 
 $last_\lab(\cE) = \bot$ if no event labelled $\lab$ exists.  
 We lift
 $last$ to sets of actions by
 $last_A(\cE) = \{last_\lab(\cE) \mid \lab \in \Act\}$.
 An event structure $\cE$ is {\em sequential} if all events are flow-ordered: $\forall e,e' \in E, e \neq e': e \causal^+ e' \vee e' \causal^+ e$. We let $\Seq$ be the set of sequential event structures.

  An event structure describes (several) concurrent executions in compact form. One execution is therein given as a configuration. 
 
 \begin{definition} 
  A {\em configuration $C \subseteq E$} of an event structure $\cE= (E,\causal,\conflict,\lres,\lfun)$ satisfies the following properties: 
 \begin{itemize*}[label={}]
     \item (1) $C$ is cycle-free: $(\causal \cap (C \times C))^+$ is irreflexive,
     \item (2) $C$ is conflict-free: $\# \cap (C \times C) = \emptyset$, 
     \item (3) $C$ is left-closed up to conflicts: $\forall d,e \in E$, if $e \in C$, $d \causal e$ and $d \notin C$, then there exists $f \in C$ such that $d \# f$ and $f \causal e$. 
 \end{itemize*}  
 \end{definition}

\noindent We let $\conf(\cE)$ be the set of configurations of $\cE$.  We identify a configuration with the (conflict-free) event structure $\cE_C$ which is $\cE$ restricted to events of $C$.  
Our intention is to use event structures to record information about the local history of each thread, in particular  the promises of other threads which they have read from. Eventually (i.e., when combining local event structures) all promises read from need to be fulfilled. This is captured by our notion of parallel composition which requires fulfills (of a thread $\tid$) to {\em synchronize} with reads from promises of $\tau$.  Similary to CCS~\cite{CCS}, we model this synchronisation via {\em complementary} actions where $\overline{\prm{\tau}{x}{\val}} = \ff{\tau}{x}{\val}$ and vice versa, and $\overline{\overline{a}} = a$. Contrary to CCS, the synchronisation does not create internal actions, but keeps the fulfill labels (as to still see what promise a fulfill belonged to). 

We first define  the {\em synchronising events} of $n$ event structures $\cE_1$,   \dots, $\cE_n$, as follows, where $E_{i*}$ denotes $E_i \cup \{*\}$.
 $$sync(\cE_1, \dots, \cE_n  ) \sdef 
 \begin{array}[c]{@{}l@{}}
 \left\{
 \begin{array}[c]{@{}l}
      (e_1, e_2, \dots, e_n) \in E_{1*} \times E_{2*}\times \dots \times E_{n*} \mid {} \\
      \exists i.\ \lfun_i(e_i) \in \Act^\ffkw \wedge (\forall j \neq i.\ e_j \neq * \imp \overline{\lfun_i(e_i)} = \lfun_j(e_j)) \wedge {}  \\  
      \qquad (\exists j \neq i.\ e_j \neq *)  
 \end{array}
 \right \}\\
      {} \cup \{(e_\inikw^1, e_\inikw^2, \dots, e_\inikw^n)\}
\end{array}      $$

An event $e$ might also occur unsynchronized in a parallel composition (which is then written as $(*,\dots,*, e,*,\dots,*)$. 

Note that since we aim to reason about reachability of states
(underapproximation), we just need parallel composition for {\em
  conflict-free} event structures, i.e.~for event structures
describing a single execution. Thus the $\Delta$-axiom of Castellani
and Zhang~\cite{CastellaniZ97} which they impose in order to get
compositionality is trivially fulfilled for our application. Next, we
still first of all define parallel composition of arbitrary event
structures.

We let $\times_i S$ denote the product $S \times S \times \dots S$ generating a tuple of length $i$. If $i \leq 0$, we let $\times_i S = \bot$. Finally, we let $\bot \times S = S \times \bot = S$. 

\begin{definition}[Parallel composition]
 Let $\cE_1, \cE_2, \dots, \cE_n$ be event structures for threads $\tid_1, \tid_2$, \dots, $\tid_n$, respectively.   
 The {\em parallel composition} $\cE = \cE_1 || \cE_2 || \dots || \cE_n$ is the event structure $(E,\causal, \#, \lres, \lfun)$ with 
   \begin{itemize}
    \item $\begin{array}[t]{@{}rcl@{}} 
        E & = & sync(\cE_1,\cE_2, \dots \cE_n)  \cup \big(\bigcup_i (\times_{i-1} \{*\}) \times (E_i \setminus \{e^i_\inikw\}) \times (\times_{n-i} \{*\})  \big)
      \end{array}$
     \item $(e_1,e_2, \dots, e_n) \causal (d_1,d_2, \dots, d_n)$ iff $\exists i.\ e_i \causal_i d_i$,
     \item $\lres((e_1,e_2, \dots, e_n), (d_1,d_2, \dots, d_n)) = \bigcup_i \lres(e_i, d_i)$, 
     \item $(e_1,e_2, \dots, e_n) \conflict (d_1,d_2, \dots, d_n) $  iff 
     \begin{itemize}
         \item $\exists i.\ e_i \conflict_i d_i$, or  \hfill (inherit conflicts)
         \item $\exists i, j.\ e_i = d_i \wedge e_i \neq * \wedge e_j \neq d_j$  \hfill (conflicts on differently paired events),   
     \end{itemize}   
   \item Labels: 
     \begin{small}
     \[
       \lfun(e_1, e_2, \dots e_n) =
       \begin{cases}
         \inikw & \text{if $(e_1, e_2, \dots e_n) = (e_{\inikw}^1, e_{\inikw}^2,
           \dots e_{\inikw}^n)$} \\
         \lfun(e_i) & \text{if $(e_1, e_2, \dots e_n) \in sync(\cE_1, \cE_2, \dots \cE_n) \wedge \lfun(e_i) = \ff{\cdot}{\cdot}{\cdot}$} \\
         \lfun(e_i) & \text{if $(e_1, e_2, \dots e_n) \notin sync(\cE_1, \cE_2, \dots \cE_n) \wedge e_i \neq *$} 
       \end{cases}
     \]
   \end{small}

 \end{itemize}
 \end{definition}

 \noindent Parallel composition of event structures is used to combine local
 proof outlines of threads. 
 This combination is only possible if
 enough synchronization partners are available.  Event structures
 $\cE_1$ to $\cE_n$ are {\em synchronizable} if
 $\pi_{i}(sync(\cE_1, \ldots, \cE_n))$ $\supseteq \Pr (E_{i})$,
 $i\in \{1, \ldots, n\}$ (all the reads have a synchronization
 with a fulfill).  The configuration (describing an execution of the
 parallel composition of threads) which we extract from
 $\conf(\cE_1 || \ldots || \cE_n)$ furthermore has to guarantee that
 no events from the local proof outlines are lost and that the local
 assertions make no contradictory assumptions about the contents
 of memory.
 
 \begin{definition}
 The event structure $\cE_C = (E_C,  \causal_C, \emptyset, \lres_C, \lfun_C)$ corresponding to a configuration $C \in \conf(\cE_1 || \ldots || \cE_n)$ is {\em interference free} if 
 \begin{enumerate}
     \item $C$ is {\em thread-covering}: $\forall i \in \{1, \ldots, n\}: \pi_i(E_C) = E_i$,   
     \item $C$ is {\em memory-consistent linearizable}: there exists a total order $\mathord{\prec} \subseteq \Act^x(E_C) \times \Act^x(E_C)$ among reads, fulfills and the $\inikw$ event such that 
     \begin{itemize}
         \item $\mathord{\causal_C}^+ \cap (\Act^x(E_C) \times \Act^x(E_C)) \subseteq \mathord{\prec}$ and 
         \item $\forall d,e,f \in E_C$: $d \rdflow{L} f \wedge d \prec e \prec f \implies e.loc \notin L$,  
     \end{itemize}
     \item $C$ contains {\em no unsynchronised reads}: there is no event in $E_C$ of the form $(*, *, \dots,*, e_i, * \dots, *)$, where $e_i \in \Pr(\cE_i)$.
 \end{enumerate}
 \end{definition}
 
 

\begin{example} \label{ex1} Consider the two event structures given next (which belong to a message passing program with barriers, see \cref{sec:reasoning}).  

\begin{minipage}{.4\textwidth}
 \begin{tikzpicture}[anchor=west,baseline,node distance=.6cm]
  \node (1) {$e_\inikw: \inikw$};
  \node [right = of 1] (2) {$e_1: \ff{1}{x}{5}$};
  \node [below = of 1] (3) {$e_2: \fence_1$};
  \node [right = of 3] (4) {$e_3: \ff{1}{y}{1}$};
  \path[->>,draw=mycolor,thick] (1) -- (2);
  \path[->>,draw=mycolor,thick] (2) -- (3); 
  \path[->>,draw=mycolor,thick] (3) -- (4);
\end{tikzpicture}  
\end{minipage} \hfill
\begin{minipage}{.6\textwidth}
   \begin{tikzpicture}[anchor=west,baseline,node distance=.4cm,draw=mycolor]
  \node (1) {$f_\inikw: \inikw$};
  \node [above right = 0.1cm and 0.4cm of 1] (2) {$f_1: \prm{1}{f}{1}$};
  \node [right = of 2] (2a) {$f_2: \barrx{a}{y}$};
  \node [below = of 2] (4) {$f_3: \barrx{b}{x}$};
  \path[->>,draw=red,thick] (1) -- node[midway,sloped,above] {\color{red} \small{\{x\}}} (2.west);
  \path[->>,draw=mycolor,thick] (2) -- (2a); 
  \path[->>,draw=mycolor,thick] (2a) -- (4.east);
  \draw[->>,thick] (1) -- (4.west);
\end{tikzpicture}
\end{minipage} 

\smallskip 
\noindent Their parallel composition gives the following event structure: 

\smallskip 
\scalebox{0.9}{
\begin{tikzpicture}[anchor=north,baseline,node distance=.4cm,draw=mycolor]
  \node (1) {$(e_\inikw, f_\inikw): \inikw$};
  \node [above right = 0.1cm and 0.6cm of 1] (1*) {$(e_1, *):\ff{1}{x}{5}$};
  \node [right = of 1*] (2*) {$(e_2, *) : \fence_1$};
  \node [right = of 2*] (3*) {$(e_3, *) : \ff{1}{y}{1}$};
  \node [below = of 2*]  (31) {$(e_3, f_1): \ff{1}{y}{1}$};
  \node [below = of 31]  (*1) {$(*, f_1): \prm{1}{y}{1}$};
  \node [right = of *1] (*2) {$(*, f_2): \barrx{a}{y}$};
  \node [below = of *1] (*3) {$(*, f_3): \barrx{b}{x}$};
  \path[->>,draw=mycolor,thick] (1) -- (1*.west);
  \path[->>,draw=mycolor,thick] (1*) -- (2*); 
  \path[->>,draw=mycolor,thick] (2*) -- (3*); 
  \path[->>,draw=mycolor,thick] (*1) -- (*2); 
  \path[->>,draw=mycolor,thick] (1) -- (*3.west); 
  \path[->>,draw=mycolor,thick] (*2) -- (*3);   
  \path[->>,draw=mycolor,thick] (2*) -- (31); 
  \path[->>,draw=mycolor,thick] (31) -- (*2); 
  \path[->>,draw=red,thick] (1) -- node[midway,sloped,above] {\color{red} \small{$\{x\}$}} (*1.west);  
  \path[->>,draw=red,thick] (1) -- node[midway,sloped,above] {\color{red} \small{$\{x\}$}} (31.west);

  \begin{scope}[draw=orange,
              every node/.style={fill=white,circle},
              every edge/.style={decorate,decoration={zigzag,segment length=2mm, amplitude=1mm},thick},
              every path/.style={decorate,decoration={zigzag,segment length=2mm, amplitude=1mm},thick}]
    \draw (3*) -- (31);
    \draw (*1) -- (31);
    \end{scope}
\end{tikzpicture}}

\smallskip 
\noindent This event structure has no interference-free configuration. To satisfy the conditions ``thread-covering'' and ``no unsynchronised reads'', we must include event $(e_3, f_1)$. This means the only possible configuration must also include the downclosure $(e_1, *)$ and $(e_2, *)$. However, together with the location restriction {\color{red} $\{x\}$} on the edge $((e_\inikw, f_\inikw), (e_3, f_1))$, the resulting event structure is not memory-consistent linearizable, since it contains a sequence $(e_\inikw, f_\inikw) \causal  (e_1, *) \causal (e_2, *) \causal (e_3, f_1)$, where $(e_1, *)$ corresponds to a fulfilled write on $x$ that is forbidden by the edge $((e_\inikw, f_\inikw) \rdflow{\{x\}} (e_3, f_1))$. Conceptually, this means that we cannot find a memory $M$ which matches the constraints on its contents given in the event structure.

\end{example}

\section{Reasoning}
\label{sec:reasoning}
Our overall objective is the design of a proof calculus for reasoning about the {\em reachability} of certain final states of concurrent programs. 
A concurrent program state describes the values of registers and shared variables, the contents of memory and the views of threads.
During reasoning, we employ event structures as {\em assertions} in proof outlines. 
They abstract from the concrete state in neither giving the exact contents of memory nor the timestamps of thread views. 

\subsection{Semantics of Assertions} 
\label{sec:meaning}

Local assertions in the proof outlines of single threads take the form $\cE$, where $\cE$ is a conflict-free event structure (i.e., $\# = \emptyset$). The event structure is conflict-free because it describes a {\em single} execution of the thread (reachability logic).  
An assertion for a thread $\tid$ can have fence and fulfill events of $\tid$, read events reading from (promises of) threads $\tid' \neq \tid$ as well as bar and test events over registers of $R(\tid)$. 
The events in $\cE$ -- together with some memory $M$ -- allow us to compute the current views of threads. 
Figure~\ref{fig:priors} gives some definitions for   calculating views. 

A local assertion of a thread $\tid$ defines
constraints on the global memory $M$ (the ordering of writes and
their values) as well as the views of $\tid$: An assertion $\cE$
describes a set of states
$ \sem{\cE} = \{\langle \vec{ts}, M \rangle \in (\TId \rightarrow {\it TState}) \times {\it Memory} \mid \langle \vec{ts}, M
\rangle \text{ matches } \cE \}$ where ``matches'' is defined by conditions {\bf (1)}-{\bf (4)} below.
 
\begin{figure}[t] 
\begin{eqnarray*}
   \priorFnc_\tau(\cE) & = & \{ e \in \Pr(\cE) \cup \Ff(\cE) \cup \{ e_\inikw \} \mid \exists e' \in last_{\fence_\tau}(\cE): e \causal^+ e' \} \\
    \priorBar_a(\cE) & = & \{ e \in \Pr(\cE) \cup \Ff(\cE) \cup \{ e_\inikw \} \mid \exists e' \in last_{\barrx{a}{\cdot}}(\cE): e \causal^+ e'  \}  \\
    \priorBar_\tau(\cE) & = &  \bigcup_{a \in R(\tau)} \priorBar_a(\cE)
   \\
    \priorBar^x_\tau(\cE) & = &    \priorBar_\tau(\cE) \cap \Act^x(\cE) \\
   \priorTest_\tau(\cE) & = &\{ e \in \Pr(\cE) \cup \Ff(\cE) \cup \{ e_\inikw \} \mid \exists e' \in last_{\test{\tid}{\cdot}}(\cE):  e \causal^+ e'  \} 
\end{eqnarray*}
\vspace{-2em}

 \caption{Determining the decisive reads and writes prior to an event ($\cE = (E,\causal,\#,\lres,\lfun)$ event structure, $\tid \in \TId$, $\reg \in \Reg, x \in \Loc$)}
 \label{fig:priors} 

\end{figure} 
 
%
      \noindent {\bf (1)} $M$ is consistent with the fulfill and read events of $\cE$.  \\
       There exists a total mapping $\psi: \Ff(\cE)  \cup \Pr(\cE) \cup \{e_{\inikw}\} \rightarrow dom(M)$ which  
       \begin{enumerate}
           \item {\em initializes at zero}: the one event $e_{\inikw}$ labelled $\inikw$ is mapped to 0, 
           \item is {\em consecutive} for every thread $\tid$: \\
           for all $e \in \Ff_\tid(\cE)$, $t \in \mathbb{T}$ s.t.~$M(t) = \langle x:=\val \rangle_{\tau}$, $t < \psi(e)$ and $e.loc = x$, there exists $d \in \Ff_\tid(\cE)$ such that $\psi(d) = t$, 
           \item {\em preserves content}: if $\psi(e) = t \neq 0$ and $M(t) =\langle x:=\val \rangle_{\tid}$, then $\lfun(e) \in \{ \ff{\tid}{x}{\val},\prm{\tid}{x}{\val} \}$, 
           \item {\em preserves flows}: $\forall e,e' \in dom(M): e \causal^+_\cE e' \imp \psi(e) < \psi(e')$,  
           
           \item and {\em preserves memory constraints}: \\
           $\forall d,e \in dom(M)$, $L \subseteq \Loc$ s.t.~$d \rdflow L e$, $\forall t \in \mathbb{T}$ s.t.~$\psi(d) < t < \psi(e)$: $M(t).loc \neq d.loc$. 
       \end{enumerate}
       The mapping $\psi$ is used to assign timestamps to read and fulfill events. We therefore will later also talk about the {\em timestamp of an event} (depending on such a mapping). 
       Note that the event structure $\Ini$ is consistent with all memories $M$ (using mapping $\psi(e_\inikw)=0$).

       \noindent {\bf (2)} The open (non-fulfilled) promises of a thread $\tau$ are
       the entries of $\tau$ in $M$ which are not fulfilled, i.e.,
       $\vec{ts}(\tau).prom = M_\tau \setminus \psi(\Ff_\tau(\cE))$.

       \noindent {\bf (3)} The views of a thread $\tau$ are consistent with mapping $\psi$ and $M$. \\  
    Letting $ts = \vec{ts}(\tau)$, $a \in R(\tid)$ and $x \in \Loc$, we have    
     \begin{align*}
           ts.v_C & = \bigsqcup_{e \in \priorTest_\tau(\cE)  } \psi(e) & 
           ts.coh(x) & = \bigsqcup_{e \in \Ff_\tau^x(\cE) \cup \priorBar_\tid^x(\cE)} \psi(e) \\
           ts.v_{\mathit{wOld}} &= \bigsqcup_{e \in \Ff_\tau(\cE)} \psi(e)   
           & 
           ts.v_{\mathit{wNew}} &= \bigsqcup_{e \in \priorFnc_\tau(\cE) \cap \big(\Ff_\tau(\cE) \cup \priorBar_\tau(\cE)\big)} \psi(e)   
        \\z
           ts.v_{a} &= \bigsqcup_{e \in \priorBar_a(\cE)} \psi(e)  & 
           ts.v_{read} &= \bigsqcup_{e \in
            \big(\priorFnc_\tau(\cE) \cap \Ff_\tid(\cE)\big) \cup \priorBar_\tau(\cE)} \psi(e) 
       \end{align*}

       \noindent {\bf (4)} The values of registers $R(\tid)$ of thread $\tid$ agree with values in $\cE$.  \\  
      For $a \in \Reg$, $ts.regs(a) = \val @ v_a$ with  $\val = \sem{a}_\cE$ (where the semantics of a register $a$ in $\cE$ is (1) 0 if no bar event for $a$ is in $\cE$ or (2)  the value of a read or fulfill to $x$ prior to the last $\barrx{a}{\cdot}$ (on $x$) or (3) the value of the expression $\expr$ in a last $\barrexp{a}{\expr}$) and $v_a$ as defined above. 
      




\begin{figure}[t] 
    \centering
    \begin{minipage}{.45\textwidth}
    \scalebox{.75}{\begin{tikzpicture}
       \coordinate (A) at (0,0);
       \coordinate (B) at (7.5,0);
       \coordinate (C) at (7.5,1);
       \coordinate (D) at (0,1);
       \coordinate (1a) at (1,0); \coordinate (1b) at (1,1);
       \coordinate (2a) at (2.5,0); \coordinate (2b) at (2.5,1); 
       \coordinate (3a) at (3.5,0); \coordinate (3b) at (3.5,1); 
       \coordinate (4a) at (5,0); \coordinate (4b) at (5,1);
       \coordinate (5a) at (6,0); \coordinate (5b) at (6,1);
       \draw [thick] (A) -- (B) --node [left,xshift=-.4cm] {$\ldots$}  (C) 
           -- node[above,xshift=-3.3cm]{0} node[above,xshift=-.8cm]{6} node[above,xshift=1.6cm]{9}(D) -- (A); 
       \draw [thick] (1a) -- node [left,xshift=.05cm] {\begin{tabular}{c} $\langle x,0 \rangle$ \\ $\langle y,0 \rangle$\end{tabular}} (1b); 
       \draw[thick] (2a) -- node [left,xshift=-.4cm] {$\ldots$}  (2b); 
       \draw[thick] (3a) -- node [left,xshift=.1cm] {$\langle y,1 \rangle_1$}(3b); 
       \draw[thick] (4a) -- node [left,xshift=-.4cm] {$\ldots$} (4b);
       \draw[thick] (5a) -- node [left,xshift=.1cm] {$\langle z,3 \rangle_1$} (5b); 
       \draw [thick,decoration={brace,mirror, raise=0.2cm},decorate] (1a.east) -- (2a) 
                    node [pos=0.5,anchor=north,yshift=-0.5cm] {\begin{tabular}{c} no write to $x$ \\
                    in $M(1)$ to $M(5)$ \end{tabular} }; 
    \end{tikzpicture}}
    \end{minipage} \hfill 
    \begin{minipage}{.3\textwidth}
        \scalebox{.8}{\begin{tikzpicture}[anchor=west,baseline,node distance=.6cm]
     \node (1) {$\inikw$};
     \node (2) [above right of=1,xshift = 1.3cm] {$\ff{1}{y}{1}$};
     \node (2a) [right of =2,xshift=1.3cm] {$\barrx{a}{y}$};
     \node (4) [below right of=1,xshift=1.3cm] {$\barrx{b}{x}$};
     \path[->>,draw=red] (1) -- node[midway,sloped,above] {\color{red} \small{\{x\}}} (2.west);
     \path[->>,draw=mycolor] (2) -- (2a); 
     \path[->>,draw=mycolor] (2a) -- (4.east); 
     \draw[->>,draw=mycolor] (1) -- (4.west);
   \end{tikzpicture} }
    \end{minipage} \hfill 
    \begin{minipage}{.18\textwidth}
       $\begin{array}[t]{l}
     \textbf{Thread } 1 \\
     1: \store{y}{1} ; \\
     2: \load{a}{y} ;\\ 
     3: \load{b}{x}; \\ 
     4: \store{z}{3}; 
     \end{array}$
   \end{minipage}
    \caption{Example memory $M$ (left) for  event structure $\cE$ (middle) describing an execution of statements 1, 2 and 3 in thread 1 (right). \\
    State $ts$ of thread 1: $prom=\{9\}, v_C=\vwnew=coh(z)=0, v_a=v_b=coh(y)=coh(x)=\vwold = v_{read}=6$, using mapping $\psi: \inikw \mapsto 0, \ff{1}{y}{1} \mapsto 6$. }
    \label{fig:ex-state}
\end{figure}

  Figure~\ref{fig:ex-state} gives an example for the definition of ``matches''. On the right hand side we see the program of 
thread 1. It first stores 1 to $y$, then loads the values of $y$ and $x$ into registers $a$ and $b$, respectively, and finally stores 3 to $z$. The event structure in the middle gives the assertion reached after statement 3, i.e.~{\em before} the final store operation. 
The memory $M$ on the left hand side matches this event structure: There are promises for the event $\inikw$ at $M(0)$ as well as for event $\ff{1}{y}{1}$, so $\psi$ maps $\inikw$ to 0 and $\ff{1}{y}{1}$ to 6. The colored location restriction in the event structure furthermore requires not to have any promises to $x$ in between 0 and 6. As there is one more promise of thread 1 in $M$, not yet covered by the event structure, we can derive $1.prom = \{9\}$.


\subsection{Proof Rules} 

Essentially, assertions describe the events which have already happened together with their orderings plus further constraints.  
\begin{figure}[t] 
\footnotesize
\begin{align*}
\cE \oplus \ff{\tid}{x}{\val} & =  (E_e,  \mathord{\causal} \cup \{ (e',e) \mid e' \in last_{\Act^x \cup \{\fence_\tid,\test{\tid}{\cdot}\}}(\cE)\}, \lres, 
    \lfun[e \mapsto \ff{\tid}{x}{\val})] ) 
    \\
\cE \oplus^a \ff{\tid}{x}{\val} & =
                                  \begin{array}[t]{@{}l@{}}
                                    (E_e,  \mathord{\causal} \cup \{ (e',e) \mid e' \in last_{\Act^x \cup \{\fence_\tid,\test{\tid}{\cdot},\barrx{a}{\cdot}\}}(\cE)\},\qquad\qquad \\
    \hfill \lres, 
    \lfun[e \mapsto \ff{\tid}{x}{\val}] )
                                  \end{array}
\\
    \cE \oplus \barrx{a}{x} & =  (E_e, \mathord{\causal} \cup \{ (e',e) \mid e' \in last_{\Act^x \cup \{ \fence_\tid,\barrx{\cdot}{\cdot} \}}(\cE) \},\lres,\lfun[e \mapsto \barrx{a}{x}]) 
    \\ 
    \cE \oplus \barrexp{a}{\expr} & =  (E_e, \mathord{\causal} \cup \left\{ 
    \begin{array}[c]{@{}l@{}}
    (e',e) \mid \\
    \exists b \in R(\expr) \cup \{a\}: e' = last_{\barrx{b}{\cdot}}(\cE) 
    \end{array}
    \right\},\lres,\lfun[e \mapsto \barrx{a}{\expr}]) 
    \\ 
    \cE \oplus \fence_\tau & =  (E_e, \mathord{\causal} \cup \{ (e',e) \mid e' \in last_{\Act \setminus \{ \test{\tid}{\cdot} \}}(\cE) \},\lres, \lfun[e \mapsto \fence_\tau]) \\
    \cE \oplus \test{\tid}{\bexp} & =  (E_e, \mathord{\causal} \cup \{ (e',e) \mid \exists a \in R(\bexp).\ e' \in last_{\barrx{a}{\cdot}}(\cE)\},\lres,\lfun[e \mapsto \test{\tid}{\bexp}]) \\
    \cE \oplus \cE' & =  \left(
    \begin{array}[c]{@{}l}
    E \cup E', \\
    \causal \cup \causal' \cup \left\{ 
    \begin{array}[c]{@{}l}
(e,e') \in E \times E' \mid \exists x \in \Loc.\ \\  
     e = last_{\{ \fence, \barrx{\cdot}{\cdot},\ff{\cdot}{x}{\cdot}\}}(\cE) \wedge 
     \lfun(e') \in \Act^x 
     \end{array}
     \right\}, 
     \\
 \lres \cup \lres', \lfun \cup \lfun'
    \end{array} \right)
\end{align*}
  \vspace{-1em}
    \caption{Operations for adding events to a conflict-free event
      structure $\cE=(E,\causal,\lres,\lfun)$, where $e \notin E$ is a
      fresh event and $E_e = E \cup \{e\}$,
      $\cE'= (E',\causal',\lres',\lfun')$, $E \cap E' = \emptyset$)}
\label{fig:plus-operations} 
\end{figure} 
The initial assertion in proof outlines is always the event structure $\Ini$.
Then,  the proof rules successively add new events to the event structure when e.g.~reading from or writing to shared variables. 
We however {\em never} add events for promises; rather, threads can first of all assume arbitrary promises of other threads having been made which they can read from. 
The overall interference freedom constraint guarantees that these local assumptions about promises are met at the end. 

For adding new events, we use a number of $\oplus$-operators, detailed in \cref{fig:plus-operations}.  
The event structures in there are local to threads and describe a single execution of the thread, hence are conflict-free. 
The definition of these operators has to ensure that they capture the dependencies between views as defined by the operational semantics. For example, rule {\sc Fulfill} requires (among others) the timestamp $t$ to be larger than control view $v_C$, hence ${\cE} \oplus \ff{\tid}{x}{\val}$ has to introduce a flow from the last test event to the newly added fulfill event.

\begin{figure}[t]\footnotesize
\begin{mathpar}
   \inferrule[\textsc{PR-Write}]{ }
      {{\color{mycolor}[\cE] } \ \store{x}{\val} \  {\color{mycolor}[  \cE \oplus \ff{\tau}{x}{\val} ] } }
      \quad
\inferrule[\textsc{PR-WriteR}]{ \sem{a}_\cE = \val   }
      {{\color{mycolor}[\cE] } \ \store{x}{a} \  {\color{mycolor}[  \cE \oplus^a \ff{\tau}{x}{\val} ] } }
      \quad 
      \inferrule[\textsc{PR-Fence}] {\ } { {\color{mycolor}[ \cE ]}  \ \kwc{dmb}_\tid \ {\color{mycolor}[ \cE \oplus \fence_\tau ]} }
      \and
      \inferrule[\textsc{PR-ReadEx}]{
        \inarr{
          e=last_{\Act^x}(\cE) \\ \lfun(e) \in \{\prm{\tau'}{x}{\val},\ff{\tau}{x}{\val}, \inikw \}  }}
      { {\color{mycolor} [\cE ] } \ \load{a}{x} \  {\color{mycolor}[ \addc_e^x(\cE \oplus  \barrx{a}{x})} ]} 
\quad
\inferrule[\textsc{PR-ReadNew}]
{
  \inarr{
    \cE' =(E',\causal',\lres',\lfun') \in \Seq \\ \lfun'(E') \subseteq \Act^{Rd} \setminus \Act_\tid \\
    last_{\Act}(\cE') = e, \lfun'(e) = \prm{\tau'}{x}{\val}
  }
}
{ {\color{mycolor}[ \cE ]} \ \load{a}{x} \  {\color{mycolor}[  (\cE \oplus \cE') \oplus \barrx{a}{x} ]} }  
      \and
\inferrule[\textsc{PR-Registers}] {  } 
     { {\color{mycolor}[ \cE  ]} \ a:= \expr\ {\color{mycolor}[ \cE \oplus  \barrexp{a}{\expr} ] } }
     \quad 
     \inferrule[\textsc{PR-Assume}] { \sem{\bexp}_\cE = \true } { {\color{mycolor}[ \cE  ]} \ \assume \ \bexp \ {\color{mycolor}[ \cE \oplus  \test{\tid} {\bexp} ] } }
 \quad 
 \inferrule[\textsc{PR-Choice}]{{\color{mycolor}[ \cE_1  ]} \ S_i \ {\color{mycolor}[ \cE_2  ]}}
     { {\color{mycolor}[ \cE_1  ]} \ S_1 + S_2 \ {\color{mycolor}[ \cE_2  ]} } 
\and
\inferrule[\textsc{PR-Sequencing}]{{\color{mycolor}[ \cE_1  ]} \ S_1 \ {\color{mycolor}[ \cE_2  ]}\qquad 
 {\color{mycolor}[ \cE_2  ]} \ S_1 \ {\color{mycolor}[ \cE_3  ]}}
 {{\color{mycolor}[ \cE_1  ]} \ S_1; S_2 \ {\color{mycolor}[ \cE_3  ]}} 
\quad
\inferrule[\textsc{PR-IterateZero}]{\ }{{\color{mycolor}[ \cE  ]} \ S^0 \ {\color{mycolor}[ \cE  ]}}
    \quad
    \inferrule[\textsc{PR-IterateNonZero}]{n > 0 \qquad {\color{mycolor}[ \cE_1  ]} \ S; S^{n-1} \ {\color{mycolor}[ \cE_2  ]}}{{\color{mycolor}[ \cE_1  ]} \ S^{n} \ {\color{mycolor}[ \cE_2  ]}}
  \end{mathpar}
    \caption{Local proof rules for a thread $\tid$}
    \label{fig:local}
\end{figure}

\Cref{fig:local} gives the proof rules for building local proof outlines of threads. 
Most of the rules (i.e., {\sc PR-Write, PR-WriteR, PR-Fence, PR-Registers} and {\sc PR-Assume}) just add one new event to the event structure recording the occurrence of a particular program statement. More complex are the two read rules: {\sc PR-ReadEx} is applied for load statements reading from $x$ when the event structure already contains an event $e$ describing (in the sense of $\sem{\cE}$) the entry in memory to read from; this can be a read, fulfill or the $\inikw$ event. In this case, the event structure after the load has to reflect the applicability condition of rule {\sc Read}: no entries in memory to $x$ in between $t$ (the timestamp of $e$ in $\sem{\cE}$) and $\vread \sqcup coh(x)$. This is achieved by inserting an additional location restriction ${\color{red}{x}}$ via the operator $\addc_e^x(\cE)$ to the following (potentially already $L$-labelled) flows (thus getting the restriction ${\color{red} L \cup \{x\}}$):   
\[ \{    e \rdflow{L} e' \mid e' \in (\priorFnc_\tau(\cE) \cap \Ff_\tau(\cE)) \cup \priorBar_\tau(\cE) \cup \Ff_\tau^x(\cE)    \} \ . \]  

Rule {\sc PR-ReadNew} on the other hand introduces new read events into an event structure upon a load statement. The rule can directly introduce an entire {\em sequence} of read events (i.e., add a sequential event structure $\cE'$) as to enable later reads from memory entries which are prior to the entry of the current read (described by event $e$ in the rule). This is required for message passing idioms like in the following program.  

\begin{center} \footnotesize
   \begin{tabular}[b]{@{}c@{}} 
   
   $\begin{array}{l@{\ \ }||@{\ \ }l}
     \begin{array}[t]{l}
      \textbf{Thread } 1 
     \\
      1: \store{x}{5} ;
      \\
      2: \code{dmb} ;
      \\
      3: \store{y}{1} ;
      \end{array}
     & 
     \begin{array}[t]{l}
      \textbf{Thread } 2 
      \\
      4: \load{a}{y} ;
      \\
      5: \load{b}{x}; 
      \end{array}      
      \end{array}$ \\
   \end{tabular} 
 \end{center} 

Here, due to the fence in Thread~1, Thread~2 -- after having read $y$ to be 1 -- can only read $x$ to be 5. When constructing the proof outline for Thread~2, we need to apply rule {\sc PR-ReadNew} for the first load giving us 
\[
\color{mycolor}{\inikw \causal \prm{1}{x}{5} \causal \prm{1}{y}{1} \causal \barrx{a}{y}}
\]
\noindent as assertion after statement 4. For the subsequent load we can then apply proof rule {\sc PR-ReadEx}. Note that we could also construct a local proof outline having the load in line 4 read from $\inikw$. This would then give us the two event structures of Example~\ref{ex1} which we, however, have already seen to not allow for an interference free configuration of their parallel composition.  

Finally, we have a proof rule for parallel composition which combines local event structures when they are synchronisable and the resulting configuration is interference free.   
 \[ \inference[\textsc{Parallel}]{\forall i \in \{1, \ldots, n\}.\ {\color{mycolor}[\Ini ]} \  S_i \ {\color{mycolor}[\cE_i ]}  \qquad \qquad \cE_1, \ldots, \cE_n \text{ synchronisable } \\ \text{ interference free } C \in \conf(\cE_1 || \ldots || \cE_n)  }
             { {\color{mycolor}[\Ini ]} \ S_1 || \ldots || S_n \ {\color{mycolor}[\cE_C ]}} \] 
             
\noindent This rule ensures that (1) all synchronization constraints are met (i.e., the promises that threads want to read from have been made)
and (2) there is a configuration $C$ of the combined event structure which is interference free. 


\begin{example} Next, we give a complete proof outline for the message passing litmus test {\em without} a barrier in the writing thread.  We see that here message passing is not
  guaranteed (i.e., reading $y$ to be 1 does not ``pass the message'' that $x$ is 5 from Thread~1 to 2) and we can actually reach a final state with
  ${\color{light-green} \big(a=1 \wedge b=0\big)}$ (as calculated by $\sem{a}_\cE$ and $\sem{b}_\cE$ taking the value of the last fulfill or $\inikw$ event prior to the last bar event on $a$ and $b$, respectively). 
    \begin{center}
    \scalebox{0.9}{
  \begin{tabular}[b]{@{}c@{}} 
   
  $\begin{array}{l@{\ \ }||@{\ \ }l}
    \begin{array}[t]{l}
     \textbf{Thread } 1 
     \\
      {\color{mycolor} \big[ \inikw \big] }  \\ 
     1: \store{x}{5} ;
     \\
    {\color{mycolor}\big[ \inikw \causal \ff{1}{x}{5} \big]  } \\
     2: \store{y}{1} ;
     \\ 
     {\color{mycolor}\left[ \begin{tikzpicture}[anchor=center,baseline,node distance=.4cm]
     \node (1) {$\inikw$};
     \node (2) [above right = -.1cm and 0.5cm of 1] {$\ff{1}{x}{5}$};
     \node (3) [below right = -.1cm and 0.5cm of 1] {$\ff{1}{y}{1}$};
     \path[->>] (1) edge (2);
     \path[->>] (1) edge (3);
   \end{tikzpicture}  
     \right]  }  
     \end{array}
    & 
    \begin{array}[t]{l}
     \textbf{Thread } 2 
     \\
     {\color{mycolor}\big[ \inikw \big] } \\
     4: \load{a}{y} ;
     \\
      {\color{mycolor}\big[ \inikw \causal \prm{1}{y}{1} \causal \barrx{a}{y}  \big] }
    \\
     5:  \load{b}{x}; 
     \\
    {\color{mycolor} \left[ \begin{tikzpicture}[anchor=center,baseline,node distance=.5cm]
     \node (1) {$\inikw$};
     \node (2) [above right = -.1cm and .5cm of 1] {$\prm{1}{y}{1}$};
     \node (2a) [right = of 2] {$\barrx{a}{y}$};
     \node (4) [below right = -.1cm and .5cm of 1] {$\barrx{b}{x}$};
     \path[->>,draw=red] (1) -- node[midway,sloped,above] {\color{red} \small{$\{x\}$}} (2.west);
     \path[->>,draw=mycolor] (2) -- (2a); 
     \path[->>,draw=mycolor] (1) -- (4.west);
   \end{tikzpicture}  \right] } 
     \end{array}      
     \end{array}$ \\
    ${\color{mycolor} \hspace*{-4em}\left[\begin{tikzpicture}[anchor=center,baseline,pos=0.5,node distance=.5cm]
     \node (1) {$\cE: \inikw$};
     \node (2) [above right = -.1cm and .5cm of 1] {$\ff{1}{y}{1}$};
     \node (2a) [right of =2,xshift=1.5cm] {$\barrx{a}{y}$};
     \node (4) [right = of 1] {$\barrx{b}{x}$};
     \node (3) [below right = -.1cm and .5cm of 1] {$\ff{1}{x}{5}$};
     \path[->>,draw=red] (1) -- node[midway,sloped,above] {\color{red} \small{$\{x\}$}} (2.west);
     \path[->>,draw=mycolor] (2) -- (2a); 
     \path[->>,draw=mycolor] (1) -- (3.west);
     \path[->>,draw=mycolor] (1) -- (4.west);
   \end{tikzpicture} \right] } $ \\
   \end{tabular} }
\end{center}

\end{example}

\subsection{Soundness and Completeness}

Due to lack of space, we can neither discuss soundness nor completeness of our proof calculus in some more detail here. Proofs can be found in the appendix.  

Soundness requires proving all local proof rules correct plus showing the correctness of rule \textsc{Parallel} as of Theorem~\ref{thm:soundness} below.  
It states that whenever we find an interference free configuration in the parallel composition of synchronizable event structures in a locally sound proof outline,  all thread states and memory contents matching this configuration are actually reachable by the concurrent program.

\begin{restatable}{theorem}{soundness}
\label{thm:soundness}
Let ${\color{mycolor}[\Ini ]} \  S_i \ {\color{mycolor}[\cE_i ]}$, $i \in \{1, \ldots, n \}$,  be proof outlines of threads $\tau_1$ to $\tau_n$ such that $\cE_1$ to $\cE_n$ are synchronizable and 
let $\vec{T_0}$ be an initial thread pool with $\vec{T_0}(\tid_i) = (S_i,ts_\inikw)$ and  $M_0 = M_\inikw$. 

Then, for every thread pool $\vec{T}$   with $\vec{T}(\tid_i) = (\skipst,ts_i)$,    interference free configuration $C \in \conf(\cE_1 || \ldots || \cE_n)$  and memory $M$   such that $\langle ts_i,M \rangle \in \sem{\cE_C}$, $tids(M) = \{\tau_1, \ldots, \tau_n\}$ and $ts_i.prom = \emptyset$, $i \in \{1, \ldots, n \}$, we have $ \langle \vec{T_0}, M_0 \rangle \trans {}^* \langle \vec{T},M \rangle$. 
\end{restatable}

Our second main result is the {\em completeness} of the proof calculus: whenever there is an execution of a concurrent program, our proof calculus allows to show the reachability of its final state. More specifically, for every trace of a concurrent program we find local proof outlines with synchronizable event structures and an interference free configuration describing the final state of the trace.

\begin{restatable}{theorem}{completeness}
  \label{th:complete}
Let $\langle \vec{T_0}, M_0 \rangle \trans {}^* \langle \vec{T},M \rangle$ be a trace of a concurrent program over threads $\tau_1, \ldots, \tau_n$ such that $\vec{T_0}$ is the initial thread pool with $\vec{T_0}(\tid_k) = (S_k,ts_\inikw)$,   $M_0 = M_\inikw$ and 
$\vec{T}$ the final thread pool with $\vec{T}(\tid_k) = (\skipst,ts_k)$ and $ts_k.prom = \emptyset$, $k\in \{1, \ldots, n\}$. 

Then there are local proof outlines ${\color{mycolor}[\Ini ]} \  S_k \ {\color{mycolor}[\cE_k ]}$   of threads $\tau_k$, $k\in \{1, \ldots, n\}$, such that $\cE_1$ to $\cE_n$ are synchronizable and there exists an interference free configuration $C \in \conf(\cE_1 || \ldots || \cE_n)$   with $\langle \vec{T},M \rangle \in \sem{\cE_C}$. 
\end{restatable}

 \section{Related Work}

The first semantics of weak memory models employing {\em promises} has been proposed by Kang et al.~in 2017~\cite{KangHLVD17} for building an operational semantics which allows modelling of read-write reordering while at the same time disallows out-of-thin-air behaviours. 
Our semantics here is a slightly simplified version of the promising semantics of ARMv8 given by Pulte et al.~\cite{PultePKLH19}. In particular, like~\cite{PultePKLH19} all program traces can be reordered so that the promise steps are all at the beginning  which is a key property required for the soundness of our proof calculus.  

There are already several proposals for program logics for weak memory
e.g.~\cite{DBLP:journals/tocl/DohertyDDW22,DBLP:conf/oopsla/VafeiadisN13,DBLP:conf/fm/CoughlinWS21,DBLP:conf/vstte/Ridge10,DBLP:conf/vmcai/DokoV16,DBLP:conf/icalp/LahavV15,DBLP:conf/popl/AlglaveC17,DBLP:conf/ppopp/DohertyDWD19,DalvandiDDW19}.
The only one explicitly dealing with promises in the semantics is the
proposal of Svendsen et
al.~\cite{DBLP:conf/esop/SvendsenPDLV18}. 
They develop a safety proof calculus whereas we are interested in
reachability.  Their logic furthermore has to deal with promises
occurring at any program step (as they show soundness with respect to
the promising semantics of~\cite{KangHLVD17}), whereas we rely on all
promises being made at the
beginning.


Partial order models of concurrency have already been used for giving
the semantics of memory
models~\cite{JagadeesanJR20,JeffreyRBCKP22,10.1145/3290383}, but not
for reasoning. \new{Wright et al~\cite{DBLP:conf/fm/WrightBD21} take
  the approach of using a {\em semantic dependency} relation, which is
  a partial order generated through an event structure representation
  of a C/C++ program~\cite{DBLP:conf/esop/PaviottiCPWOB20}, which is a
  partial order over a thread's execution. An Owicki-Gries logic is
  provided to reason directly over such partial orders.  }
Incorrectness logic as used for proving reachability properties of {\em sequential} programs has been introduced by O'Hearn~\cite{OHearn20}, with a predecessor approach with (almost) the same principles by de Vries and Koutavas~\cite{DBLP:conf/sefm/VriesK11}. The first extension of incorrectness logic to concurrent programs has been proposed by Raad et al.~in the form of an incorrectness separation logic~\cite{RaadBDO22} which is however not compositional.

\new{ Colvin~\cite{DBLP:conf/sefm/Colvin21} defines a semantics based
  on a reordering relation for several hardware memory models, which
  is then lifted to a Hoare calculus. This is then rephrased into a
  reachability property by defining triples
  $\langle\!\langle p \rangle\!\rangle ~s~\langle\!\langle q
  \rangle\!\rangle = \neg \{p\}~s~\{\neg q\}$, which states that it is
  possible for $s$ to reach $q$ if execution starts in a state
  satisfying $p$. Note that this is weaker than O'Hearn's notion of
  incompleteness, which states that all states satisfying $q$ are reachable
  from an execution starting in a state satisfying $p$.}


 
\section{Conclusion}

In this paper, we have proposed a reachability (incorrectness) logic for concurrent programs running on weak memory models. 
The reasoning technique is based on assertions which are event structures abstractly describing the contents of memory and the views of all threads. We have proven soundness and completeness of the proof calculus, and have demonstrated its applicability by proving the outcomes of some standard litmus tests to be reachable. 


\smallskip
\noindent {\bf Acknowledgements.} We thank Christopher Pulte for clarifying one aspect of the Register rule of ARMv8's operational semantics to us and Sadegh Dalvandi for initial discussions on the semantics.

\bibliographystyle{splncs04}
\bibliography{references}

\appendix
\section{Further examples}

\paragraph*{RRC -- Read Read Coherence}   
In the read read coherence litmus test the first two threads store different values to $x$, while the other two threads load $x$ into different registers. The outcome ${\color{red} \big(a=1 \wedge b=2 \wedge c=2 \wedge d=1\big)\FAIL}$ should not be allowed and indeed, the parallel composition of local proof outlines does not give us an interference-free configuration.  
Here, we also give the event names in the assertions.  

\medskip 
\begin{center} 
\hspace*{-30pt}\scalebox{0.95}{
\begin{tabular}[b]{@{}l@{}} 
   
  $\begin{array}{l@{\ \ }||@{\ \ }l||@{\ \ }l||}
    \begin{array}[t]{l}
     \textbf{Thread } 1 
     \\
      {\color{mycolor} \big[ e_\inikw: \inikw \big] }  \\ 
     \store{x}{1} ;
     \\
     {\color{mycolor} \big[ e_\inikw: \inikw \causal e_1:\ff{1}{x}{1}  \big] } \\
     
     \end{array}
    & 
    \begin{array}[t]{l}
     \textbf{Thread } 2 
     \\
      {\color{mycolor} \big[ e_\inikw:  \inikw \big] }  \\ 
     \store{x}{2} ;
     \\
     {\color{mycolor} \big[ e_\inikw: \inikw \causal e_2:\ff{2}{x}{2}  \big] } \\
     \end{array}
    \\
    \multicolumn{3}{l}{\ }
    \\
    \begin{array}[t]{l}
     \textbf{Thread } 3 
     \\
      {\color{mycolor} \big[ e_\inikw:  \inikw \big] }  \\ 
     \load{a}{x} ;
     \\
     {\color{mycolor} \big[e_\inikw:  \inikw \causal e_3:\prm{1}{x}{1} \causal e_4:\barrx{a}{x} \big] } \\
     \load{b}{x} ;
     \\
     {\color{mycolor} \left[  \begin{tikzpicture}[anchor=south,baseline,node distance=.6cm]
     \node (1) {$e_3:\prm{1}{x}{1}$} ;  
     \node (0) [left =0.3cm of 1] {$e_\inikw: \inikw$};
     \node (12) [right =0.3cm of 1] {$e_4:\barrx{a}{x}$};
     \node (2) [below of=1] {$e_5:\prm{2}{x}{2}$};
     \node (22)[below of=12] {$e_6:\barrx{b}{x}$};
     \path[->>] (0) edge (1);
     \path[->>] (12) edge (2);
     \path[->>] (1) edge (12);
     \path[->>] (2) edge (22);   
   \end{tikzpicture}  \right]  }  
   \\
     \end{array}
    & 
    \multicolumn{2}{@{}l}{
    \begin{array}[t]{l}
     \textbf{Thread } 4 
     \\
      {\color{mycolor} \big[ e_\inikw:  \inikw \big] }  \\ 
     \load{c}{x} ;
     \\
     {\color{mycolor} \big[e_\inikw:  \inikw \causal e_7:\prm{2}{x}{2} \causal e_8:\barrx{c}{x} \big] } \\
     \load{d}{x} ;
     \\
     {\color{mycolor} \left[  \begin{tikzpicture}[anchor=south,baseline,node distance=.6cm]
     \node (1) {$e_7:\prm{2}{x}{2}$} ;  
     \node (0) [left =0.3cm of 1] {$e_\inikw: \inikw$};
     \node (12) [right =0.3cm of 1] {$e_8:\barrx{c}{x}$};
     \node (2) [below of=1] {$e_9:\prm{1}{x}{1}$};
     \node (22)[below of=12] {$e_{10}:\barrx{d}{x}$};
     \path[->>] (0) edge (1);
     \path[->>] (12) edge (2);
     \path[->>] (1) edge (12);
     \path[->>] (2) edge (22);   
   \end{tikzpicture}  \right]  } 
   \\
     \end{array}}
     
   \end{array}$ 
   \\
   \qquad\qquad\qquad\qquad\qquad\qquad\qquad ${\color{red} \big(a=1 \wedge b=2 \wedge c=2 \wedge d=1 \big) \FAIL }$ 
   \end{tabular} 
   }
\end{center}

\medskip 
\noindent 
  Here, the parallel composition of the local assertion event structures contains a cycle

\begin{center} 
\scalebox{0.9}{\begin{tikzpicture}
\begin{scope}[every node/.style={rounded corners,thick}]
    \node (0)  {$(e_\inikw,e_\inikw,e_\inikw, e_\inikw)$};
    \node [right=of 0](1)  {$(*,*,*, e_8)$};
    \node [right=of 1] (12) {$(e_1, *,e_3,e_9)$};
    \node [right=of 12] (13)  {$(*,*, e_4,*)$};
    \node [below=of 1] (2)  {$(*,*,*,e_{10})$};
    \node [below=of 12] (22) {$(*,e_2, e_5,e_7)$};
    \node [below=of 13] (23) {$(*,*, e_6,*)$};
\end{scope}

\begin{scope}[
              every node/.style={fill=white},
              every edge/.style={->>,draw,thick}]
    \path (0) edge (1);
    \path (0) edge (2);
    \path (1) edge (12);
    \path (12) edge (13);
    \path (12) edge  (2);
    \path (13) edge  (22);
    \path (22) edge  (1);
    \path (22) edge  (23);
    \end{scope}
\end{tikzpicture}}
\end{center} 
\noindent and therefore has no interference free configuration.

\paragraph*{WRC -- Write-to-Read Causality}

The next litmus test is the WRC example, originally formulated by Boehm and Adve~\cite{DBLP:conf/pldi/BoehmA08} and appearing here in the form of \cite{tutorial}.

\begin{center} 
\scalebox{0.95}{
\begin{tabular}[b]{@{}l@{}} 
   
  $\begin{array}{l@{\ \ }||@{\ \ }l||@{\ \ }l}
    \begin{array}[t]{l}
     \textbf{Thread } 1 
     \\
      {\color{mycolor} \big[  \inikw \big] }  \\ 
     1: \store{x}{1} ; \\ 
      {\color{mycolor} \big[  \inikw \causal \ff{1}{x}{1}  \big] }  \\ 
     \end{array}
    & 
    \begin{array}[t]{l}
     \textbf{Thread } 2 
     \\
     {\color{mycolor} \big[  \inikw \big] }  \\ 
     2: \load{a}{x} ;
     \\
     {\color{mycolor} \big[  \inikw \causal \prm{1}{x}{1} \causal \barrx{a}{x} \big] }  \\ 
     3: \store{y}{1} ; \\
      {\color{mycolor} \left[  \begin{tikzpicture}[anchor=south,baseline,node distance=.6cm]
     \node (0) {$\inikw$};
     \node (1) [above right of=0,xshift=1.1cm] {$\prm{1}{x}{1}$} ;  
     \node (12) [right =0.3cm of 1] {$\barrx{a}{x}$};
     \node (2) [below right of=0,xshift=1.1cm] {$\ff{2}{y}{1}$};
     \path[->>] (0) edge (1);
     \path[->>] (1) edge (12);
     \path[->>] (0) edge (2);   
   \end{tikzpicture}  \right]  } 
   \\
     \end{array}
    & 
    \begin{array}[t]{l}
     \textbf{Thread } 3 
     \\
     {\color{mycolor} \big[  \inikw \big] }  \\  
     4: \load{b}{y} ;
     \\
      {\color{mycolor} \big[  \inikw \causal \prm{2}{y}{1} \causal \barrx{b}{y} \big] }  \\
     5: \load{c}{x} ;
     \\
      {\color{mycolor} \left[ \begin{tikzpicture}[anchor=west,baseline,node distance=.4cm]
     \node (1) {$\inikw$};
     \node (2) [above right of=1,xshift = 1.5cm] {$\prm{2}{y}{1}$};
     \node (2a) [right = of 2] {$\barrx{b}{y}$};
     \node (4) [below of =2] {$\barrx{c}{x}$};
     \path[->>,draw=red] (1) -- node[midway,sloped,above] {\color{red} \small{$\{x\}$}} (2.west);
     \path[->>,draw=mycolor] (2) -- (2a); 
     \path[->>,draw=mycolor] (2a) -- (4.east);
     \path[->>,draw=mycolor] (1) -- (4.west);
   \end{tikzpicture}  \right] } 
     \end{array}
     
   \end{array}$ 
   \\ \qquad \qquad \qquad \quad 
    ${\color{mycolor} \left[ \begin{tikzpicture}[anchor=south,baseline,pos=0.4,node distance=.6cm]
     \node (1) {$\inikw$};
     \node (2) [above right of=1,xshift = 1.5cm] 
     {$\ff{2}{y}{1}$};
     \node (2a) [right = of 2] {$\barrx{b}{y}$};
     \node (3) [below right of=1,xshift = 1.5cm,yshift=-.4cm] {$\ff{1}{x}{1}$};
     \node (3b) [right of=3, xshift=1.5cm] {$\barrx{a}{x}$};
     \node (4) [below of=2] {$\barrx{c}{x}$};
     \path[->>,draw=red] (1) -- node[midway,sloped,above] {\color{red} \small{$\{x\}$}} (2.west);
     \path[->>,draw=mycolor] (2) -- (2a); 
     \path[->>,draw=mycolor] (3) -- (3b); 
     \path[->>,draw=mycolor] (1) -- (3.west);
     \path[->>,draw=mycolor] (2a) -- (4.east);
     \path[->>,draw=mycolor] (1) -- (4.west);
   \end{tikzpicture}  \right] } $ \\
   \qquad\qquad\qquad \qquad \quad ${\color{light-green} \big(a=1 \wedge b=1 \wedge c=0 \big) \OK}$ 
   \end{tabular} }
\end{center} 

\noindent We see that due to the lack of a barrier in between statements 2 and 3 in Thread 2, the final outcome ${\color{light-green} \big(a=1 \wedge b=1 \wedge c=0\big)\OK }$ is reachable. 

\paragraph*{SB -- Store buffering}
Next, we take a look at the standard SB example. We see that the outcome ${\color{light-green} \big(a=0 \wedge b=0 \big)\OK }$ which is possible in lots of weak memory models is also possible here, and can be proven to be reachable. 

\begin{center} 
  \begin{tabular}[b]{@{}c@{}} 
   
  $\begin{array}{l@{\ \ }||@{\ \ }l}
    \begin{array}[t]{l}
     \textbf{Thread } 1 
     \\
     {\color{mycolor} \big[  \inikw \big] }  \\
     1: \store{x}{1} ;
     \\
     {\color{mycolor} \big[  \inikw \causal \ff{1}{x}{1}  \big] }  \\ 
     2: \load{a}{y} ;
     \\ 
     {\color{mycolor} \left[ \begin{tikzpicture}[anchor=west,baseline,node distance=.6cm]
     \node (1) {$\inikw$};
     \node (2) [above right of=1,xshift = 1.5cm] {$\ff{1}{x}{1}$};
     \node (2a) [below right of =1,xshift=1.5cm] {$\barrx{a}{y}$};
     \path[->>,draw=mycolor] (1) -- (2a.west); 
     \path[->>,draw=mycolor] (1) -- (2.west);
   \end{tikzpicture}  \right] } 
     \end{array}
    & 
    \begin{array}[t]{l}
     \textbf{Thread } 2 
     \\
     {\color{mycolor} \big[  \inikw \big] }  \\
     3: \store{y}{1} ;
    \\
    {\color{mycolor} \big[  \inikw \causal \ff{2}{y}{1}  \big] }  \\ 
    4: \load{b}{x}; 
    \\
    {\color{mycolor} \left[ \begin{tikzpicture}[anchor=west,baseline,node distance=.6cm]
     \node (1) {$\inikw$};
     \node (2) [above right of=1,xshift = 1.5cm] {$\ff{2}{y}{1}$};
     \node (2a) [below right of =1,xshift=1.5cm] {$\barrx{b}{x}$};
     \path[->>,draw=mycolor] (1) -- (2a.west); 
     \path[->>,draw=mycolor] (1) -- (2.west);
   \end{tikzpicture}  \right] } 
  
     \end{array}      
     \end{array}$ \\
     ${\color{mycolor} \left[ \begin{tikzpicture}[anchor=west,baseline,node distance=.6cm]
     \node (1) {$\inikw$};
     \node (2) [above right of=1,xshift = 2cm] {$\barrx{a}{y}$};
     \node (3) [below right of =1,xshift=2cm] {$\ff{2}{y}{1}$};
     \node (4) [above of=2] {$\ff{1}{x}{1}$};
     \node (5) [below of=3] {$\barrx{b}{x}$};
     \path[->>,draw=mycolor] (1) -- (2.west); 
     \path[->>,draw=mycolor] (1) -- (3.west);
     \path[->>,draw=mycolor] (1) -- (4.west);
     \path[->>,draw=mycolor] (1) -- (5.west);
   \end{tikzpicture}  \right] } $ \\
   
       \quad ${\color{light-green} \big(a=0 \wedge b=0 \big) \OK }$ 
   \end{tabular} 
\end{center} 
   
\paragraph*{IRIW -- Independent Reads of Independent Writes}
In the IRIW litmus test we have two threads writing to shared variables $x$ and $y$, respectively, and two other threads both reading both variables, but in a different order. 
The outcome ${\color{red} \big(a=1 \wedge b=0 \wedge c=1 \wedge d=0\big)\FAIL}$ is not allowed by our semantics. 

\medskip 
\scalebox{0.95}{
   \begin{tabular}[b]{@{}l@{}} 
  $\begin{array}{l@{\ \ }||@{\ \ }l||@{\ \ }l||@{\ \ }l}
    \begin{array}[t]{l}
     \textbf{Thread } 1 
     \\
     {\color{mycolor} \big[ \inikw  \big] } \\
     \store{x}{1} ;
     \\
     {\color{mycolor} \big[ \inikw \causal \ff{1}{x}{1}  \big] } \\
     
     \end{array}
    & 
     \begin{array}[t]{l}
     \textbf{Thread } 2
     \\
     {\color{mycolor} \big[ \inikw  \big] } \\
     \load{a}{x} ;
     \\
     {\color{mycolor} \big[ \inikw \causal \prm{1}{x}{1} \causal \barrx{a}{x} \big] } \\
     \load{b}{y} ;
     \\
     {\color{mycolor} \left[  \begin{tikzpicture}[anchor=south,baseline,node distance=0.6cm]
     \node (1) {$\inikw$} ;  
     \node (2) [above right of=1,xshift = 1.2cm] {$\prm{1}{x}{1}$};
     \node (3) [right of=2, xshift = 1.2cm] {$\barrx{a}{x}$};
     \node (4) [below right of=1,xshift = 1.2cm] {$\barrx{b}{y}$};
     \path[->>,draw=red] (1) -- node[midway,sloped,above] {\color{red} \small{$\{y\}$}} (2.west);
     \path[->>] (2) edge (3);
     \path[->>] (1) edge (4.west);
     \path[->>] (3) edge (4.east);   
   \end{tikzpicture}  \right]  }  
   \\
     \end{array}
    \\
        \multicolumn{3}{l}{\ }
    \\

    \begin{array}[t]{l}
     \textbf{Thread } 3
     \\
     {\color{mycolor} \big[ \inikw  \big] } \\
     \store{y}{1} ;
     \\
     {\color{mycolor} \big[ \inikw \causal \ff{3}{y}{1}  \big] } \\
     \end{array}
    & 
       \multicolumn{2}{@{}l}{
    \begin{array}[t]{l}
     \textbf{Thread } 4 
     \\
     {\color{mycolor} \big[ \inikw  \big] } \\
     \load{c}{y} ;
     \\
     {\color{mycolor} \big[ \inikw \causal \prm{3}{y}{1} \causal \barrx{c}{y} \big] } \\
     \load{d}{x} ;
     \\
     {\color{mycolor} \left[  \begin{tikzpicture}[anchor=south,baseline,node distance=0.6cm]
     \node (1) {$\inikw$} ;  
     \node (2) [above right of=1,xshift = 1.2cm] {$\prm{3}{y}{1}$};
     \node (3) [right of=2, xshift = 1.2cm] {$\barrx{c}{y}$};
     \node (4) [below right of=1,xshift = 1.2cm] {$\barrx{d}{x}$};
     \path[->>,draw=red] (1) -- node[midway,sloped,above] {\color{red} \small{$\{x\}$}} (2.west);
     \path[->>] (2) edge (3);
     \path[->>] (1) edge (4.west);
     \path[->>] (3) edge (4.east);   
   \end{tikzpicture}  \right]  } 
   \\
     \end{array}}
     
   \end{array}$
   \\
   \qquad\qquad ${\color{red} \big(a=1 \wedge b=0 \wedge c=1 \wedge d=0\big) \FAIL }$ 
   \end{tabular} }
 
 \medskip   
   And indeed, while we can construct local event structures for each thread,  the one candidate for a configuration of their parallel composition is
   
 \begin{center}
 ${\color{mycolor} \left[  \begin{tikzpicture}[anchor=south,baseline,node distance=0.8cm]
     \node (1) {$\inikw$} ;  
     \node (2) [above right of=1,xshift = 1.5cm] {$\barrx{b}{y}$};
     \node (3) [below right of=1,xshift = 1.5cm] {$\ff{3}{y}{1}$};
     \node (4) [above of= 2] 					 {$\ff{1}{x}{1}$};
     \node (5) [below of= 3] 					 {$\barrx{d}{x}$};
     \node (6) [right of= 3, xshift = 1.5cm] 	 {$\barrx{c}{y}$};
     \node (7) [right of= 4, xshift = 1.5cm] 	 {$\barrx{a}{x}$};
     \path[->>,draw=red] (1) -- node[midway,sloped,above] {\color{red} \small{$\{y\}$}} (4);
     \path[->>,draw=red] (1) -- node[midway,sloped,above] {\color{red} \small{$\{x\}$}} (3);
     \path[->>] (1) edge (2);
     \path[->>] (1) edge (5);
     \path[->>] (3) edge (6); 
     \path[->>] (4) edge (7); 
     \path[->>] (7) edge (2); 
     \path[->>] (6) edge (5); 
   \end{tikzpicture}  \right]  } $.
 \end{center} 
 
 But this configuration is not interference free (because of the location restrictions there is no valid linearization of the read and fulfill events), and we hence (correctly) cannot show outcome ${\color{red} \big(a=1 \wedge b=0 \wedge c=1 \wedge d=0\big) \FAIL }$ to be reachable.  
\section{Proofs of Soundness and Completeness}

\subsection{Properties of the Operational Semantics}

We start with two properties of the operational semantics. The first ensures that in any trace, the promises can be reordered so that they are executed at the start of the trace without affecting the final outcome. 
\begin{lemma} \label{lemma:promises-first}
  For every trace $tr$ there exists a trace $tr'$ (ending in the same final state) such that in $tr'$ all promise actions of $tr$ occur at the beginning  followed by the non-promise actions of $tr$ (in the same order as in $tr$).  
\end{lemma}
\begin{proof}
Let 
$\langle \overrightarrow{T_0}, M_0 \rangle \trans {op_1} \langle \overrightarrow{T_1}, M_1 \rangle \trans {op_2} \ldots \trans {op_{k-1}} \langle \overrightarrow{T_k}, M_k \rangle$ 
be a trace $tr$ with $I=\{ i \mid op_i = \prm{\cdot}{\cdot}{\cdot}\}$ being the set of all promises. 
Let $i_1<i_2<...$ be the elements of $I$. 
We now want to swap  $op_{i_1}$ with the operation before it, until $op_{i_1}$ is the first one in the trace. 
After that we do the same with $op_{i_2}$ until it is the second operation and so on. 
For that we need to show that we can swap a promise with a fence, fulfill or read operation before it. 
\\\\
Let $op_i$ be a promise operation, $l$ its location and $op_{i-1}$ the operation in front of $op_i$. 
If we promise something, we change $prom$ to $prom \cup \{|M_{i-1}|+1\}$ in $ts$ and add one element to the memory. 
In case that $op_{i-1}$ is a fence operation we can reorder $op_i$ and $op_{i-1}$, because a fence operation changes different attributes of $ts$ ($v_{read}$ and $v_{write}$) and has no effects on the memory.
\\\\
If $op_{i-1}$ is a fulfill operation, it fulfills a different promise. 
Hence the condition 
$$v_{write} \sqcup ts.coh(l) \sqcup v_{\reg}  < t $$
still applies after the reordering, because $t$ depends on the promise $op_{i-1}$ is fulfilling. 
The changes of $coh$ and $v_{write}$ are not related to the changes caused by $op_i$. 
Additionally the element deleted by $op_{i-1}$ in $prom$ is not the same as the element that $op_i$ adds. 
So we can also reorder a promise and a fulfill. 
\\\\
In case that $op_{i-1}$ is a read operation, $op_{i-1}$ changes other attributes of $ts$ then $op_i$ would. 
If $op_{i-1}$ reads from a location $l'$, which is different from $l$, the condition 
$$\forall t'.\ t < t' \le (ts.v_{read} \sqcup ts.coh(l)) \imp M(t').loc \neq l $$
holds. 
However, when $op_{i-1}$ reads from $l$, we need to have a deeper look at $v_{read}$ and $coh$. 
Both will only be replaced by timestamps of entries in $M_{i-1}$. 
Those timestamps are always smaller then $|M_{i-1}|+1$, so the condition is still true. 
Hence we can reorder a read and a promise operation. 
\\\\
That means we can reorder the operations of our execution in a way that all promises are at the beginning and the final state stays the same. 
\end{proof}

The second property ensures that views of each component monotonically increase after each transition. Note that this differs from the semantics of Pulte et al., where register views might decrease after the execution of a write to a local register~\cite{PultePKLH19}.

\begin{lemma}[View monotonicity]\label{lemma:views-monotonic} 
Suppose 
  $\langle\langle S, ts \rangle, M \rangle \trans {op}_{\tid} \langle\langle S', ts' \rangle, M' \rangle$. Then for $view \in \{v_{read},v_{wNew}, v_{wOld}, v_C\}$, we have 
  $ts.view \leq  ts'.view$. Moreover, for all locations $x$, we have $ts.coh(x) \leq ts'.coh(x)$, and for all registers $a$, we have that if  $ts.regs(a)= \val @ v$ and $ts'.regs(a) = \val' @ v'$, then $v \leq v'$. 
\end{lemma} 
\begin{proof}
Let $\langle\langle S, ts \rangle, M \rangle \trans {op}_{\tid} \langle\langle S', ts' \rangle, M' \rangle$. 
\begin{itemize}
	\item $v_C$ only changes for $op=asm(bexp)$: $ts'.v_C=ts.v_C \sqcup ts.v \geq ts.v_C$
	\item $v_{wOld}$ only changes for $op={\it ff}(\loc,\val)$: $ts'.v_{wOld} = ts.v_{wOld} \sqcup t \geq ts.v_{wOld}$
	\item $coh(x)$ changes in two cases:
	\begin{itemize}
		\item for $op=rd(\loc,\val)$: $ts'.coh(\loc) = ts.coh(\loc) \sqcup v_{post} \geq ts.coh(\loc)$
		\item for $op={\it ff}(\loc,\val)$: $ts'.coh(\loc) = t > ts.coh(\loc)$ because $ts.v_{wNew} \sqcup ts.v_C \sqcup ts.coh(\loc) \sqcup v_a  < t$
	\end{itemize}
	\item $v_{read}$ changes in two cases:
	\begin{itemize}
		\item for $op=rd(\loc,\val)$: $ts'.v_{read} = v_{post} = ts.v_{read} \sqcup t \geq ts.v_{read}$ 
		\item for $op=\fence$: $ts'.v_{read}=ts.v_{read} \sqcup ts.v_{wOld} \geq ts.v_{read}$
	\end{itemize}
	\item $v_{wNew}$ only changes for $op=\fence$: $ts'.v_{wNew}=ts.v_{read} \sqcup ts.v_{wOld} \geq ts.v_{read}$. By looking at $v_{read}$ we see that $ts'.v_{read}\geq ts.v_{read}$ holds in all cases of $op$. Since $v_{read}$ changes whenever $v_{wNew}$ does and both are set to the same value we get $ts.v_{read} \geq ts.v_{wNew}$. Therefore $ts'.v_{wNew} \geq ts.v_{wNew}$.
	\item For $regs(a)=\val @ v_a$ we first show that $ts.v_a\leq ts.v_{read}$ is an invariant. We know that $v_a$ only changes for $op=rd(\loc,\val)$ and $op=\lst{a}{exp}$ whereas $v_{read}$ changes for $op=rd(\loc,\val)$ and $op=\fence$. For $op=rd(\loc,\val)$ $v_a$ and $v_{read}$ are set to the same value. Since $ts'.v_{read}\geq ts.v_{read}$, $v_a$ can not get bigger than $v_{read}$ after a register operation. Therefore our invariant holds and we get: 
	\begin{itemize}
		\item for $op=rd(\loc,\val)$: $ts'.v_a = ts.v_{read} \sqcup t \geq ts.v_{read} \geq ts.v_a$
		\item since $op=\lst{a}{exp}$ changes $v_a$ to the maximum of $v_a$ and the view of $exp$., $ts'.v_a \geq ts.v_a$
	\end{itemize}
\end{itemize}
\end{proof}

\subsection{Soundness}

Next we prove soundness of our proof rules. We first prove correctness of each of the atomic rules from \cref{fig:local}. 

\begin{proposition}[Soundness of rule \textsc{PR-WriteC}] 
  Let $\langle ts',M \rangle \in \sem{\cE \oplus \ff{\tau}{x}{\val}}$. Then there exists $\langle ts, M \rangle \in \sem{\cE }$ such that

  $\big\langle \langle \store{x}{\val}, ts \rangle, M \big\rangle \trans {\ff{\tau}{x}{\val}} \big \langle \langle \skipst, ts' \rangle, M \big \rangle$. 
\end{proposition}
\begin{proof}
  We write $\cE'$ for $\cE \oplus^\reg \ff{\tau}{x}{\val}$. Let $\langle ts', M \rangle \in \sem{\cE'}$ and let $e'$ be the event such that $\lfun(e') = \ff{\tau}{x}{\val}$ 
  This means that there exists some mapping $\psi': \Ff(\cE') \cup \Pr(\cE') \rightarrow dom(M)$ that shows  $\langle ts,M\rangle$ satisfies $\cE'$. 
  In particular, $\psi'(e') = t$ implies that $M(t) = \langle x := \val \rangle_\tid$. 
  We  show that there exists some $\langle ts,M \rangle \in \sem{\cE}$ such that the following holds (eliding the program statements):
  \[ \langle ts, M \rangle \trans {\ff{\tau}{x}{\val}} \langle ts', M \rangle \]
  Note that $M$ is unchanged and $t$ in the rule {\sc Fulfill} is instantiated to $ts'.coh(x)$. We define $ts$ to be the following by using $\psi'$: 
  \begin{align*}
     ts.prom & :=  ts'.prom \cup \{ts'.coh(x)\} & 
     ts.coh(x) & := \bigsqcup_{e \in \Ff_\tau^x(\cE) \cup \priorBar_\tau(\cE)} \psi'(e) \\
     ts.v_{wOld} & := \bigsqcup_{e \in \Ff_\tau(\cE)} \psi'(e)  & 
     ts.v_{wNew} & := ts'.v_{wNew} \\
     ts.v_\reg & :=  0 & 
      ts.v_{read} & :=  ts'.v_{read} \\
      ts.v_C & := ts'.v_C
  \end{align*}




First, we show $\langle ts,M\rangle \in \sem{\cE}$ by checking the conditions in \cref{sec:meaning}. 
\begin{enumerate}
\item We define $\psi = \{e'\} \ndres \psi'$, where $\ndres$ denotes domain anti-restriction. Clearly $M$ remains consistent with $\cE$.
\item This is trivial. We have one additional open promise in $ts.prom$, but this is precisely the fulfilled write in $\cE'$.
\item 
\begin{itemize}
    \item 
    \begin{align*}
      ts.v_C & = ts'.v_C & & \text{by definition} \\
             & = \bigsqcup_{e \in \priorTest_\tau(\cE')} \psi'(e) & & \text{unfolding definition of $ts'.v_C$} \\
             & = \bigsqcup_{e \in \priorTest_\tau(\cE)} \psi'(e) & & \text{$\priorTest_\tau(\cE') = \priorTest_\tau(\cE)$}\\ 
             & = \bigsqcup_{e \in \priorTest_\tau(\cE)} \psi(e) & & \text{$\priorTest_\tau(\cE) \subseteq dom(\psi)$}    \end{align*}

    \item 
    \begin{align*}
        ts.coh(x) & = \bigsqcup_{e \in \Ff_\tau^x(\cE) \cup \priorBar_\tau(\cE)} \psi'(e) &&\text{definition} \\
        & = \bigsqcup_{e \in \Ff_\tau^x(\cE) \cup \priorBar_\tau(\cE)} \psi(e) && \text{$\Ff_\tau^x(\cE) \cup \priorBar_\tau(\cE) \subseteq dom(\psi)$}
    \end{align*}
    \item 
    \begin{align*}
        ts.v_{wOld} & = \bigsqcup_{e \in \Ff_\tau(\cE)} \psi'(e)  && \text{definition} \\
                    & = \bigsqcup_{e \in \Ff_\tau(\cE)} \psi(e)  && 
        \text{$\Ff_\tau(\cE) \subseteq dom(\psi)$}
    \end{align*}

    \item 
      \begin{small}
    \begin{align*}
         ts.v_{wNew} & = ts'.v_{wNew} && \text{by definition}\\
                     & = \bigsqcup_{e \in \mathit{\priorFnc}_\tau(\cE') \cap \big(\Ff_\tau(\cE') \cup \priorBar_\tau(\cE')\big)} \psi(e) && \text{unfolding definition}\\
                     & = \bigsqcup_{e \in \mathit{\priorFnc}_\tau(\cE) \cap \big(\Ff_\tau(\cE') \cup \priorBar_\tau(\cE')\big)} \psi(e) && \text{since $\mathit{\priorFnc}_\tau(\cE') = \mathit{\priorFnc}_\tau(\cE)$} \\
                    & = \bigsqcup_{e \in \mathit{\priorFnc}_\tau(\cE) \cap \big(\Ff_\tau(\cE) \cup \priorBar_\tau(\cE)\big)} \psi(e) && \text{logic}
    \end{align*}
  \end{small}

    \item 
      \begin{small}
        \begin{align*}
           ts.v_{read} & = ts.v_{read} && \text{by definition}\\ 
                       & = \bigsqcup_{e \in
            \big(\mathit{\priorFnc}_\tau(\cE') \cap \Ff_\tid(\cE')\big) \cup \mathit{\priorBar}_\tau(\cE')} \psi(e) && \text{unfolding} \\
            & = \bigsqcup_{e \in
            \big(\mathit{\priorFnc}_\tau(\cE) \cap \Ff_\tid(\cE')\big) \cup \mathit{\priorBar}_\tau(\cE)} \psi(e)
            && \text{\begin{tabular}{@{}l}
            since $\mathit{\priorFnc}_\tau(\cE') = \mathit{\priorFnc}_\tau(\cE)$ \\ and $\mathit{\priorBar}_\tau(\cE') = \mathit{\priorBar}_\tau(\cE)$
            \end{tabular}}
            \\
            & = \bigsqcup_{e \in
            \big(\mathit{\priorFnc}_\tau(\cE) \cap \Ff_\tid(\cE)\big) \cup \mathit{\priorBar}_\tau(\cE)} \psi(e)
            && \text{logic}        
        \end{align*}
      \end{small}

\end{itemize}
\item Trivial since no register value is modified.
\end{enumerate}

Recall: 
\begin{align*}
    & \cE \oplus \ff{\tid}{x}{\val} \\
  = & (E \cup \{e \},  \mathord{\causal} \cup \{ (e',e) \mid e' \in last_{\Act^x \cup \{\fence_\tid,\test{\tid}{\cdot}\}}(\cE), \lres, 
    \lfun[e \mapsto \ff{\tid}{x}{\val})] ) 
\end{align*}

Next, we have $\langle ts, M \rangle \trans {\ff{\tau}{x}{\val}} \langle ts', M \rangle$ as follows:
\begin{itemize}
    \item $ts'.prom = ts.prom \backslash \{t\}$ by definition
    \item $regs(\val) = \val @0$ because $ts.v_\reg = ts'.v_\reg$ and $ts.regs(\reg) = ts'.regs(\reg)$,  
    \item $M(t) = \langle x := \val \rangle_\tid$ by definition
    \item $ts.v_{wNew} \sqcup ts.v_C \sqcup ts.coh(x)  \sqcup ts.v_\reg < ts'.coh(x)$. We prove maximality of $ts'.coh(x)$ for each component: 
        \begin{itemize}
            \item $ts.v_{wNew} < ts'.coh(x)$ 

            Recall that              
            $ts.v_{wNew} = \bigsqcup_{e \in \mathit{\priorFnc}_\tau(\cE) \cap (\Ff_\tau(\cE) \cup \priorBar_\tau(\cE))} \psi(e)  $. 
            By definition of $\cE \oplus^a \ff{\tid}{x}{\val}$, we have that the event corresponding to $\ff{\tid}{x}{\val}$ is causally after each $e \in \mathit{\priorFnc}_\tau(\cE)$. Thus,  $ts.v_{wNew} < ts'.coh(x)$. 

\item $ts.v_C < ts'.coh(x)$. 

    Recall that $ts.v_C = \bigsqcup_{e \in \priorTest_\tau(\cE)} \psi(e)$. By definition of $\cE \oplus^a \ff{\tid}{x}{\val}$, we have that the event corresponding to $\ff{\tid}{x}{\val}$ is causally after each $e \in \mathit{\priorTest}_\tau(\cE)$. Thus $ts.v_C < ts'.coh(x)$.

  \item $ts.coh(x) < ts'.coh(x)$ because by definition of $\cE \oplus^a \ff{\tid}{x}{\val}$ we have a causal edge from each $e \in Act^x$ to the event corresponding to $\ff{\tid}{x}{\val}$. 
  
  \item $ts.v_\reg < ts'.coh(x)$ because $v_\reg = v_\val = 0$. 
        \end{itemize}

\end{itemize}  \end{proof}

\begin{proposition}[Soundness of rule \textsc{PR-WriteR}] 
  Let $\langle ts',M \rangle \in \sem{\cE \oplus^\reg \ff{\tau}{x}{\val}}$ and $\sem{a}_\cE = \val$. Then there exists $\langle ts, M \rangle \in \sem{\cE }$ such that

  $\big\langle \langle \store{x}{\reg}, ts \rangle, M \big\rangle$ $\trans {\ff{\tau}{x}{\val}} \big \langle \langle \skipst, ts' \rangle, M \big \rangle$. 
\end{proposition}
\begin{proof}
  We write $\cE'$ for $\cE \oplus^\reg \ff{\tau}{x}{\val}$. Let $\langle ts', M \rangle \in \sem{\cE'}$ and let $e'$ be the event such that $\lfun(e') = \ff{\tau}{x}{\val}$ 
  This means that there exists some mapping $\psi': \Ff(\cE') \cup \Pr(\cE') \rightarrow dom(M)$ that shows  $\langle ts,M\rangle$ satisfies $\cE'$. 
  In particular, $\psi'(e') = t$ implies that $M(t) = \langle x := \val \rangle_\tid$. 
  We  show that there exists some $\langle ts,M \rangle \in \sem{\cE}$ such that the following holds (eliding the program statements):
  \[ \langle ts, M \rangle \trans {\ff{\tau}{x}{\val}} \langle ts', M \rangle \]
  Note that $M$ is unchanged and $t$ in the rule {\sc Fulfill} is instantiated to $ts'.coh(x)$. We define $ts$ to be the following by using $\psi'$: 
  \begin{align*}
     ts.prom & :=  ts'.prom \cup \{ts'.coh(x)\} & 
     ts.coh(x) & := \bigsqcup_{e \in \Ff_\tau^x(\cE) \cup \priorBar_\tau(\cE)} \psi'(e) \\
     ts.v_{wOld} & := \bigsqcup_{e \in \Ff_\tau(\cE)} \psi'(e)  & 
     ts.v_{wNew} & := ts'.v_{wNew} \\
     ts.v_\reg & :=  ts'.v_\reg & 
      ts.v_{read} & :=  ts'.v_{read} \\
      ts.v_C & := ts'.v_C
  \end{align*}




First, we show $\langle ts,M\rangle \in \sem{\cE}$ by checking the conditions in \cref{sec:meaning}. 
\begin{enumerate}
\item We define $\psi = \{e'\} \ndres \psi'$, where $\ndres$ denotes domain anti-restriction. Clearly $M$ remains consistent with $\cE$.
\item This is trivial. We have one additional open promise in $ts.prom$, but this is precisely the fulfilled write in $\cE'$.
\item 
\begin{itemize}
    \item 
    \begin{align*}
      ts.v_C & = ts'.v_C & & \text{by definition} \\
             & = \bigsqcup_{e \in \priorTest_\tau(\cE')} \psi'(e) & & \text{unfolding definition of $ts'.v_C$} \\
             & = \bigsqcup_{e \in \priorTest_\tau(\cE)} \psi'(e) & & \text{$\priorTest_\tau(\cE') = \priorTest_\tau(\cE)$}\\ 
             & = \bigsqcup_{e \in \priorTest_\tau(\cE)} \psi(e) & & \text{$\priorTest_\tau(\cE) \subseteq dom(\psi)$}    \end{align*}

    \item 
    \begin{align*}
        ts.coh(x) & = \bigsqcup_{e \in \Ff_\tau^x(\cE) \cup \priorBar_\tau(\cE)} \psi'(e) &&\text{definition} \\
        & = \bigsqcup_{e \in \Ff_\tau^x(\cE) \cup \priorBar_\tau(\cE)} \psi(e) && \text{$\Ff_\tau^x(\cE) \cup \priorBar_\tau(\cE) \subseteq dom(\psi)$}
    \end{align*}
    \item 
    \begin{align*}
        ts.v_{wOld} & = \bigsqcup_{e \in \Ff_\tau(\cE)} \psi'(e)  && \text{definition} \\
                    & = \bigsqcup_{e \in \Ff_\tau(\cE)} \psi(e)  && 
        \text{$\Ff_\tau(\cE) \subseteq dom(\psi)$}
    \end{align*}

    \item 
      \begin{small}
    \begin{align*}
         ts.v_{wNew} & = ts'.v_{wNew} && \text{by definition}\\
                     & = \bigsqcup_{e \in \mathit{\priorFnc}_\tau(\cE') \cap \big(\Ff_\tau(\cE') \cup \priorBar_\tau(\cE')\big)} \psi(e) && \text{unfolding definition}\\
                     & = \bigsqcup_{e \in \mathit{\priorFnc}_\tau(\cE) \cap \big(\Ff_\tau(\cE') \cup \priorBar_\tau(\cE')\big)} \psi(e) && \text{since $\mathit{\priorFnc}_\tau(\cE') = \mathit{\priorFnc}_\tau(\cE)$} \\
                    & = \bigsqcup_{e \in \mathit{\priorFnc}_\tau(\cE) \cap \big(\Ff_\tau(\cE) \cup \priorBar_\tau(\cE)\big)} \psi(e) && \text{logic}
    \end{align*}
  \end{small}

    \item \begin{align*}
         ts.v_{\reg} & = ts'.v_{\reg} && \text{by definition}\\
                     & = \bigsqcup_{e \in \mathit{\priorBar}_a(\cE')} \psi(e) && \text{unfolding definition}\\
                     & = \bigsqcup_{e \in \mathit{\priorBar}_a(\cE)} \psi(e) && \text{since $\mathit{\priorBar}_a(\cE') = \mathit{\priorBar}_a(\cE)$}
    \end{align*}    
    \item 
      \begin{small}
        \begin{align*}
           ts.v_{read} & = ts.v_{read} && \text{by definition}\\ 
                       & = \bigsqcup_{e \in
            \big(\mathit{\priorFnc}_\tau(\cE') \cap \Ff_\tid(\cE')\big) \cup \mathit{\priorBar}_\tau(\cE')} \psi(e) && \text{unfolding} \\
            & = \bigsqcup_{e \in
            \big(\mathit{\priorFnc}_\tau(\cE) \cap \Ff_\tid(\cE')\big) \cup \mathit{\priorBar}_\tau(\cE)} \psi(e)
            && \text{\begin{tabular}{@{}l}
            since $\mathit{\priorFnc}_\tau(\cE') = \mathit{\priorFnc}_\tau(\cE)$ \\ and $\mathit{\priorBar}_\tau(\cE') = \mathit{\priorBar}_\tau(\cE)$
            \end{tabular}}
            \\
            & = \bigsqcup_{e \in
            \big(\mathit{\priorFnc}_\tau(\cE) \cap \Ff_\tid(\cE)\big) \cup \mathit{\priorBar}_\tau(\cE)} \psi(e)
            && \text{logic}        
        \end{align*}
      \end{small}

\end{itemize}
\item Trivial since no register value is modified.
\end{enumerate}

Recall: 
\begin{align*}
    & \cE \oplus^a \ff{\tid}{x}{\val} \\
  = & (E \cup \{e \},  \mathord{\causal} \cup \{ (e',e) \mid e' \in last_{\Act^x \cup \{\fence_\tid,\test{\tid}{\cdot},\barrx{a}{\cdot}\}}(\cE), \lres, 
    \lfun[e \mapsto \ff{\tid}{x}{\val})] ) 
\end{align*}

Next, we have $\langle ts, M \rangle \trans {\ff{\tau}{x}{\val}} \langle ts', M \rangle$ as follows:
\begin{itemize}
    \item $ts'.prom = ts.prom \backslash \{t\}$ by definition
    \item $regs(\reg) = \val @v_\reg$ because $ts.v_\reg = ts'.v_\reg$ and $ts.regs(\reg) = ts'.regs(\reg)$,  
    \item $M(t) = \langle x := \val \rangle_\tid$ by definition
    \item $ts.v_{wNew} \sqcup ts.v_C \sqcup ts.coh(x)  \sqcup ts.v_\reg < ts'.coh(x)$. We prove maximality of $ts'.coh(x)$ for each component: 
        \begin{itemize}
            \item $ts.v_{wNew} < ts'.coh(x)$ 

            Recall that              
            $ts.v_{wNew} = \bigsqcup_{e \in \mathit{\priorFnc}_\tau(\cE) \cap (\Ff_\tau(\cE) \cup \priorBar_\tau(\cE))} \psi(e)  $. 
            By definition of $\cE \oplus^a \ff{\tid}{x}{\val}$, we have that the event corresponding to $\ff{\tid}{x}{n}$ is causally after each $e \in \mathit{\priorFnc}_\tau(\cE)$. Thus,  $ts.v_{wNew} < ts'.coh(x)$. 

\item $ts.v_C < ts'.coh(x)$. 

    Recall that $ts.v_C = \bigsqcup_{e \in \priorTest_\tau(\cE)} \psi(e)$. By definition of $\cE \oplus^a \ff{\tid}{x}{\val}$, we have that the event corresponding to $\ff{\tid}{x}{\val}$ is causally after each $e \in \mathit{\priorTest}_\tau(\cE)$. Thus $ts.v_C < ts'.coh(x)$.

  \item $ts.coh(x) < ts'.coh(x)$ because by definition of $\cE \oplus^a \ff{\tid}{x}{\val}$ we have a causal edge from each $e \in Act^x$ to the event corresponding to $\ff{\tid}{x}{\val}$. 
  
  \item $ts.v_a < ts'.coh(x)$ same 
  Because  
              \[
              \bigsqcup_{e \in \mathit{\priorBar}_r(\cE)} \psi(e)
                = 
              \bigsqcup_{e \in \mathit{\priorBar}_r(\cE')} \psi(e)
  \]
  Moreover $\ff{\tid}{x}{\val}$ goes at the end and $\mathit{\priorBar}_r(\cE) = \mathit{\priorBar}_r(\cE')$
        \end{itemize}

\end{itemize}  
\end{proof}

\begin{proposition}[Soundness of rule \textsc{PR-ReadEx}] \label{prop:readex}
  Let $\langle ts',M \rangle \in \sem{\addc^x_e(\cE \oplus \barrx{a}{x})  }$, where $e$ is the event in rule \textsc{PR-ReadEx}. Then there exists $\langle ts, M \rangle \in \sem{\cE  }$ such that 
  
  $\big\langle \langle \load{a}{x}, ts \rangle, M \big\rangle \trans {rd(\loc,\val)}_\tid \big \langle \langle \skipst, ts' \rangle, M \big \rangle$. 
\end{proposition}

\begin{proof}  
  For the proof, we elide the program statements. We write $\cE'$ for $\addc_e^x(\cE \oplus \barrx{a}{x})$. Let $\langle ts', M \rangle \in \sem{\cE'}$. 
  This means that there exists some mapping $\psi': \Ff(\cE') \cup \Pr(\cE') \cup \{e_{\inikw}\} \rightarrow dom(M)$ such that $\psi'$ can be used to show that $\langle ts',M\rangle$ satisfies $\cE'$, in particular by preservation of memory constraints for the location-restricted red flows for all timestamps $t'$, events $f$ such that $\psi'(e) < t' < \psi'(f)$ and $e \causal f$ and $f \in (\mathit{\priorFnc}_\tau(\cE') \cap \Ff_\tau(\cE')) \cup \mathit{\priorBar}_\tau(\cE') \cup \Ff_\tau^x(\cE')$ we have $M(t').loc \neq e.loc$. We now show that there exists some $\langle ts,M \rangle \in \sem{\cE}$ such that 
  \[ \langle ts, M \rangle \trans {rd(x,\val)}_{\tid} \langle ts', M \rangle \]
  Note that $M$ stays the same as we do not have any further promises. We define $ts$ to be the following (using $\psi:= \psi'$):  
  \begin{align*}
     ts.prom & :=  ts'.prom  & 
     ts.v_{wOld} & := ts'.v_{wOld} \\
     ts.v_{wNew} & := ts'.v_{wNew} & 
     ts.v_b & := ts'.v_b \qquad (b \neq a) \\ 
     ts.coh(x) & := \bigsqcup_{e \in \Ff_\tau^x(\cE) \cup \priorBar^x_\tau(\cE)} \psi(e) & 
     ts.v_a & :=  \bigcup_{e \in \priorBar_a(\cE)} \psi(e) \\
      ts.v_{read} & :=  \bigsqcup_{e \in
            (\mathit{\priorFnc}_\tau(\cE) \cap \Ff_\tid(\cE)) \cup \mathit{\priorBar}_\tau(\cE)} \psi(e) & 
            ts.v_C & := ts'.v_C 
  \end{align*} 
  Since $\cE'$ neither has additional fulfill nor read events nor fences nor tests and $M$ is the same and  $\priorFnc_\tau(\cE) = \priorFnc_\tau(\cE')$, we by definition get $\langle ts, M \rangle \in \sem{\cE}$. 
  Next, we need to determine the timestamp $t$ for reading, which in this case is $\psi(e)$. Since $\psi$ preserves contents and $\lfun(e) \in \{\prm{\tau}{x}{\val}, \ff{\tau}{x}{\val},\inikw\}$, we get $M(t).loc = x$ and $M(t).val = \val$ ($\val=0$ in case of reading initial values). Now consider some timestamp $t'$ such that $t < t' \leq ts.v_{read} \sqcup ts.coh(x)$. 
  Let $f$ be the event such that $\psi(f) = ts.v_{read} \sqcup ts.coh(x)$. This implies (definition of $ts.coh(x)$ and $ts_{v_{read}}$) that $f \in \Ff_\tau^x(\cE) \cup \priorBar^x_\tau(\cE) \cup (\mathit{\priorFnc}_\tau(\cE) \cap \Ff_\tid(\cE)) \cup \mathit{\priorBar}_\tau(\cE)$. Thus we also get $f \in (\mathit{\priorFnc}_\tau(\cE') \cap \Ff_\tau(\cE')) \cup \mathit{\priorBar}_\tau(\cE') \cup \Ff_\tau^x(\cE')$. By preservation of memory constraints we thus have $M(t').loc \neq e.loc$. 
  
  Next to the views. We need to show the correct afterstate for three views being changed by the read rule: $v_a, coh(x)$ and $v_{read}$. 
  \begin{itemize}
      \item $v_a$ is set to $ts.v_{read} \sqcup t$. We get 
         \begin{align*}
             ts.v_{read} \sqcup t &= \bigsqcup_{e \in (\priorFnc_\tau(\cE) \cap \Ff_\tau(\cE)) \cup \priorBar_\tau(\cE)} \psi(e) \sqcup t \\
             &= \bigsqcup_{e \in \priorBar_\tau(\cE')} \psi'(e) \\
             &= \bigsqcup_{e \in \priorBar_a(\cE')} \psi'(e) \\
             &= ts'.v_a
         \end{align*}
         as $\oplus$ in particular introduces a flow between the event read from (which is labelled with some action in $\Act^x$) and other bar events to the event $\barrx{a}{x}$. 
      \item $coh(x)$ is set to $coh(x) \sqcup v_{read} \sqcup t$. We get  
        \begin{small}
        \begin{align*}
            ts.coh(x) \sqcup ts.v_{read} \sqcup t &= \bigsqcup_{e \in \Ff_\tau^x(\cE) \cup \priorBar^x_\tau(\cE)} \psi(e) \sqcup \bigsqcup_{e \in
            (\mathit{\priorFnc}_\tau(\cE) \cap \Ff_\tid(\cE)) \cup \mathit{\priorBar}_\tau(\cE)} \psi(e) \sqcup t \\
            &= \bigsqcup_{e \in \Ff_\tau^x(\cE)} \psi(e) \sqcup \bigsqcup_{e \in \mathit{\priorBar}_\tau(\cE)} \psi(e) \sqcup t \\
            & = \bigsqcup_{e \in \Ff_\tau^x(\cE')} \psi(e) \sqcup \bigsqcup_{e \in \mathit{\priorBar}^x_\tau(\cE')} \psi(e) \\
            &= ts'.coh(x) 
        \end{align*}
      \end{small}

        This holds as the new bar event is flow after all other bars and all fences (by definition of $\oplus$). 
        \item $v_{read}$ is set to $v_{read} \sqcup t$. We get 
        \begin{align*}
             ts.v_{read} \sqcup t &= \bigsqcup_{e \in (\priorFnc_\tau(\cE) \cap \Ff_\tau(\cE)) \cup \priorBar_\tau(\cE)} \psi(e) \sqcup t \\
             &= \bigsqcup_{e \in (\priorFnc_\tau(\cE') \cap \Ff_\tau(\cE')) \cup \priorBar_\tau(\cE')} \psi(e) \\
             &= ts'.v_{read}
        \end{align*}
        \item Finally, we look at the value in register $a$ which is set to $\val$ (which is $M(t).val$). $\sem{a}_{\cE'} = \val$ as the reading event $e$ is  in $last_{\Act^x}(\cE')$ (as this is equal to $last_{\Act^x}(\cE)$ and reading from last events is required by the hypothesis of rule \textsc{PR-ReadEx}) and the event is in $\priorBar_a(\cE')$ (by definition of $\oplus$). 
  \end{itemize}
\end{proof}

 \begin{proposition}[Soundness of rule \textsc{PR-ReadNew}] \label{prop:readnew}
  Let $\langle ts',M \rangle \in \sem{(\cE \oplus \cE') \oplus \barrx{a}{x}  }$, $\cE'$ as of rule \textsc{PR-ReadNew}. Then there exists $\langle ts, M \rangle \in \sem{\cE  }$ such that  $\big\langle \langle \load{a}{x}, ts \rangle, M \big\rangle \trans {rd(x,\val)}_\tid \big \langle \langle \skipst, ts' \rangle, M \big \rangle$. 
 \end{proposition}

 \begin{proof} 
   For the proof, we elide the program statements. We write $\cE_{\mathit{after}}$ for $(\cE \oplus  \cE') \oplus \barrx{a}{x}$. Let $\langle ts', M \rangle \in \sem{\cE_{\mathit{after}}  }$. 
  This means that there exists some mapping $\psi': \Ff(\cE_{\mathit{after}}) \cup \Pr(\cE_{\mathit{after}}) \cup \{e_\inikw\} \rightarrow dom(M)$ such that $\psi'$ can be used to show that $\langle ts',M\rangle$ satisfies $\cE_{\mathit{after}}$. Let $e$ be the event in $\cE'$ which we read from and let $t=\psi'(e)$. We set $ts$ to be the thread state which we construct from $\cE$ and $M$ using $\psi$ for which we take $\psi'|_\cE$. By definition of $\cE_{\mathit{after}}$, $e$ is flow-after all fulfills of $x$, all fences plus reads in $\cE$ which are followed by a bar event (and $e$ not in $\cE$). As $\psi'$ is preserving flows, all events $f \in (\priorFnc_\tau(\cE) \cap \Ff_\tau(\cE)) \cup \priorBar(\cE)$ (which determine $ts.v_{read}$) have a timestamp less than $t$, hence $ts.v_{read} < t$. Similarly, all events $f \in \Ff_\tau^x(\cE) \cup \priorBar^x_\tau(\cE)$ (which determine $ts.coh(x)$) have a timestamp 
  less than $t$ and hence $ts.coh(x) < t$. Hence the condition of the read rule in the operational semantics about locations is trivially fulfilled. Moreover, $v_{post}$ is $t$. Thus, $v_a$ as well as $v_{read}$ and $coh(x)$ are set to $t$, and this then coincides with $ts'.v_a$,  (as $ts'.v_a = \bigsqcup _{f \in \priorBar_a(\cE_{after})} \psi'(f)$ and $e$ is the maximal event in this set), with $ts'.v_{read}$ (similarly, because $e$ is the maximal event in $(\priorFnc_\tau(\cE_{after}) \cap \Ff_\tau(\cE_{after})) \cup \priorBar(\cE_{after})$ and with $ts'.coh(x)$ 
  (similarly, because $e$ is the maximal event in $\Ff_\tau^x(\cE_{after}) \cup \priorBar_\tau^x(\cE_{after})$), respectively. Views $v_{wOld}$, $v_{wNew}$ and $v_C$ keep their values as of the operational semantics, and indeed calculating their views from $\cE$ and $\cE_{after}$ gives the same value as we neither introduce new test events (this would change $v_C$) nor new fulfills (would change $v_{wOld}$) nor fences (would change $v_{wNew}$). 
  
   Finally, we look at the value in register $a$ which is set to $\val$ (which is $M(t).val$). $\sem{a}_{\cE_{\mathit{after}}} = \val$ as the reading event $e$ with label $\prm{\tau'}{x}{\val}$ is  in $last_{\Act^x}(\cE_{\mathit{after}}) \cap \priorBar_a(\cE_{\mathit{after}})$. 
 \end{proof} 

\begin{proposition}[Soundness of rule \textsc{PR-Fence}]
  Let $\langle ts',M \rangle \in \sem{\cE \oplus \fence_\tau  }$. Then there exists $\langle ts, M \rangle \in \sem{\cE  }$ such that  $\big\langle \langle \kwc{dmb}, ts \rangle, M \big\rangle \trans {\fence}_{\tid} \big \langle \langle \skipst, ts' \rangle, M \big \rangle$. 
\end{proposition}

\begin{proof} 
  For the proof, we elide the program statements. We write $\cE'$ for $\cE \oplus \fence_\tau$. Let $\langle ts', M \rangle \in \sem{\cE'   }$. 
  This means that there exists some mapping $\psi': \Ff(\cE') \cup \Pr(\cE') \cup \{e_\inikw\} \rightarrow dom(M)$ such that $\psi'$ can be used to show that $\langle ts,M\rangle$ satisfies $\cE'$. 
  We now show that there exists some $\langle ts,M \rangle \in \sem{\cE}$ such that 
  \[ \langle ts, M \rangle \trans {\fence}_{\tid} \langle ts', M \rangle \]
  Note that $M$ stays the same as we do not have any further promises. We define $ts$ to be the following by using $\psi:= \psi'$: 
  \begin{align*}
     ts.prom & :=  ts'.prom  & 
     ts.coh(x) & := ts'.coh(x) \\
     ts.v_{wOld} & := ts'.v_{wOld} & ts.v_r & :=  ts'.v_r \\
     ts.v_{wNew} & := \bigsqcup_{e \in
            \mathit{\priorFnc}_\tau(\cE) \cap (\Ff_\tid(\cE) \cup \mathit{\priorBar}_\tau(\cE))} \psi(e)  && \\
      ts.v_{read} & :=  \bigsqcup_{e \in
            (\mathit{\priorFnc}_\tau(\cE) \cap \Ff_\tid(\cE)) \cup \mathit{\priorBar}_\tau(\cE)} \psi(e) &&
  \end{align*} 
  Since $\cE'$ neither has additional fulfill, read nor bar events, $M$ is the same and  $\priorBar_\tau(\cE) = \priorBar_\tau(\cE')$ is the same, we by definition get $\langle ts, M \rangle \in \sem{\cE}$. Next to the views. We need to show the correct afterstate for two views being changed by the fence rule: $v_{wNew}$ and $v_{read}$. Both are set to $ts.v_{read}\sqcup ts.v_{wOld}$.
  \begin{itemize}
      \item For $v_{read}$ we get 
         \begin{align*}
             ts.v_{read}\sqcup ts.v_{wOld} 
             &= \bigsqcup_{e \in (\priorFnc_\tau(\cE) \cap \Ff_\tau(\cE)) \cup \priorBar_\tau(\cE)} \psi(e) \sqcup \bigsqcup_{e \in \Ff_\tau(\cE)} \psi(e) \\
             &= \bigsqcup_{e \in (\priorFnc_\tau(\cE) \cap \Ff_\tau(\cE)) \cup \priorBar_\tau(\cE) \cup \Ff_\tau(\cE)} \psi(e) \\
             &= \bigsqcup_{e \in (\priorFnc_\tau(\cE) \cap \Ff_\tau(\cE')) \cup \priorBar_\tau(\cE') \cup \Ff_\tau(\cE)} \psi(e) \\
             &= \bigsqcup_{e \in (\priorFnc_\tau(\cE) \cap \Ff_\tau(\cE')) \cup \priorBar_\tau(\cE') \cup (\priorFnc_\tau(\cE') \cap \Ff_\tau(\cE'))} \psi(e) \\
             &= \bigsqcup_{e \in (\priorFnc_\tau(\cE') \cap \Ff_\tau(\cE')) \cup \priorBar_\tau(\cE')} \psi'(e) \\
             &= ts'.v_{read}
         \end{align*}
         as $\oplus \fence_\tau$ introduces a flow from all last events to the new fence and so $\Ff_\tau(\cE)=\priorFnc_\tau(\cE') \cap \Ff_\tau(\cE')$. 
        \item  For $v_{wNew}$ we get 
        \begin{align*}
             ts.v_{read}\sqcup ts.v_{wOld} 
             &= \bigsqcup_{e \in (\priorFnc_\tau(\cE) \cap \Ff_\tau(\cE)) \cup \priorBar_\tau(\cE) \cup \Ff_\tau(\cE)} \psi(e) \\
             &= \bigsqcup_{e \in \priorBar_\tau(\cE) \cup \Ff_\tau(\cE)} \psi(e) \\
             &= \bigsqcup_{e \in \priorBar_\tau(\cE') \cup \Ff_\tau(\cE')} \psi'(e) \\
             &= \bigsqcup_{e \in \priorFnc_\tau(\cE') \cap (\priorBar_\tau(\cE') \cup \Ff_\tau(\cE'))} \psi'(e) \\
             &= ts'.v_{wNew}
        \end{align*}
        as $\priorBar_\tau(\cE')$ and $\Ff_\tau(\cE')$ are subsets of $\priorFnc_\tau(\cE')$.
  \end{itemize}
\end{proof}

\begin{proposition}[Soundness of rule \textsc{PR-Register}] \label{prop:register}
  Suppose that $\langle ts',M \rangle \in \sem{\cE \oplus \barrx{a}{exp}  }$. Then there exists $\langle ts, M \rangle \in \sem{\cE  }$ such that

  $\big \langle \langle a := exp, ts \rangle, M \big \rangle \trans {\lst{a}{exp} }_\tau   \big \langle  \langle \skipst,ts'\rangle, M \big \rangle$. 
\end{proposition}

\begin{proof} 
  For the proof, we elide the program statements. We write $\cE'$ for $\cE \oplus \barrx{a}{exp}$. Let $\langle ts', M \rangle \in \sem{\cE'   }$. 
  This means that there exists some mapping $\psi': \Ff(\cE') \cup \Pr(\cE') \cup \mathsf{Ini} \rightarrow dom(M)$ such that $\psi'$ can be used to show that $\langle ts,M\rangle$ satisfies $\cE'$. 
  We now show that there exists some $\langle ts,M \rangle \in \sem{\cE}$ such that 
  \[ \langle ts, M \rangle \trans {\lst{a}{exp}}_{\tid} \langle ts', M \rangle \]
  Note that $M$ stays the same as we do not have any further promises. We define $ts$ to be the following by using $\psi:= \psi'$: 
  \begin{align*}
     ts.prom & :=  ts'.prom \\
     ts.coh(x) & := ts'.coh(x) \\
     ts.v_{wOld} & := ts'.v_{wOld} \\
     ts.v_{wNew} & := ts'.v_{wNew} \\
     ts.v_a & := \bigsqcup_{e \in \priorBar_a(\cE)} \psi(e) \\
      ts.v_{read} & := ts'.v_{read} 
  \end{align*} 
We need to show the correct afterstate for $v_a$ and $reg(a)$ being changed by the register rule. 
  \begin{itemize}
      \item For $v_{a}$ we get 
         \begin{align*}
             ts.v_a \sqcup \bigsqcup_{a'\in R(exp)}ts.v_{a'} 
             &= \bigsqcup_{e\in \priorBar_a(\cE)} \psi(e) \sqcup \bigsqcup_{a'\in R(exp)}\bigsqcup_{e\in \priorBar_{a'}(\cE)} \psi(e) \\
             &= \bigsqcup_{e\in \priorBar_a(\cE) \cup \bigcup_{a' \in R(exp)}\priorBar_{a'}(\cE) } \psi(e) \\
             &= \bigsqcup_{e\in \priorBar_a(\cE')} \psi'(e)
         \end{align*}
         as $\oplus \barrx{a}{exp}$ introduces for every $b\in R(exp)\cup \{a\}$ a flow from the last $\barrx{b}{\cdot}$ to $\barrx{a}{exp}$ and so $\priorBar_a(\cE')=\priorBar_a(\cE) \cup \bigcup_{a' \in R(exp)}\priorBar_{a'}(\cE)$.
	  \item For $reg(a)$ we need to show $\sem{a}_{\cE'}=\val$ with $\sem{a}_{reg}=\val@v$. Let $exp=e_1 op_1 \ldots op_{k-1} e_k$ with $\sem{e_i}_{reg}=\val_i@v_i$. We get 
	  	\begin{align*}
	  		\sem{a}_{reg} 
	  		&= \sem{exp}_{reg} \\
	  		&= \sem{e_1 op_1 \ldots op_{k-1} e_k}_{reg} \\
	  		&= (\val_1 \sem{op_1} \ldots \sem{op_{k-1}} \val_k)@(v_1 \sqcup \ldots \sqcup v_k)
		\end{align*}
		and
		\begin{align*}
			\sem{a}_{\cE'} 
	  		&= \sem{exp}_{\cE'} \\
	  		&= \sem{e_1 op_1 \ldots op_{k-1} e_k}_{\cE'} \\
	  		&= \sem{e_1}_{\cE'} \sem{op_1} \ldots \sem{op_{k-1}} \sem{e_1}_{\cE'} \\
	  		&= \val_1 \sem{op_1} \ldots \sem{op_{k-1}} \val_k
		\end{align*}
		as $\sem{e_i}_{\cE'}=\sem{e_i}_{\cE}$ and $\sem{e_i}_{\cE}=\sem{e_i}_{reg}$ for all $1\leq i\leq k$.		  	   
  \end{itemize}
\end{proof}

\begin{proposition}[Soundness of rule \textsc{PR-Assume}]\label{prop:assume}
  Suppose that $\langle ts',M \rangle \in \sem{\cE   \oplus \test{\tau}{bexp}  }$ such that $\sem{bexp}_\cE = \true$. Then there exists $\langle ts, M \rangle \in \sem{\cE  }$ such that  $\big\langle \langle \kwc{assume}\ bexp , ts \rangle, M \big\rangle \trans {asm(bexp))}_\tid \big \langle \langle \skipst, ts' \rangle, M \big \rangle$. 
\end{proposition}

\begin{proof} 
For the proof, we elide the program statements. We write $\cE'$ for $\cE \oplus \test{\tau}{bexp}$. Let $\langle ts', M \rangle \in \sem{\cE'}$. 
  This means that there exists some mapping $\psi': \Ff(\cE') \cup \Pr(\cE') \cup \{e_\inikw\} \rightarrow dom(M)$ such that $\psi'$ can be used to show that $\langle ts',M\rangle$ satisfies $\cE'$. 
  We now show that there exists some $\langle ts,M \rangle \in \sem{\cE}$ such that 
  \[ \langle ts, M \rangle \trans {\test{\tid}{bexp}}_\tid \langle ts', M \rangle \]
  Note that $M$ stays the same as we do not have any further promises. We define $ts$ to be the following by using $\psi:= \psi'$: We take $ts.v := ts'.v$ for all views except for $v_C$. For the control dependence view $v_C$ we take $ts.v_C := \bigsqcup_{e \in \priorTest_\tau(\cE)} \psi(e)$. Moreover, we let $ts.regs(a) = \val_a@ts.v_a$ where $\val_a = \sem{a}_\cE$. 
  Note that $ts \in \sem{\cE}$. 
  Furthermore, since $\sem{bexp}_\cE = \true$ (precondition of proof rule for assume), we get $\sem{bexp}_{ts.regs} = \true@v$ for $v=v_{a_1} \sqcup \ldots \sqcup v_{a_m}$ such that $R(bexp) = \{a_1, \ldots, a_m\}$. 
  Hence the transition $asm(bexp)$ is enabled in $ts$. It remains to be shown that we reach $ts'$ by taking this transition. The interesting case is $v_C$: 
  \begin{eqnarray*}
      ts.v_C \sqcup v & = & \bigsqcup_{e \in \priorTest_\tau(\cE)} \psi(e) \sqcup v \\
                      & = & \bigsqcup_{e \in \priorTest_\tau(\cE)} \psi(e) \sqcup v_{a_1} \sqcup \ldots v_{a_n} \\
                      & = & \bigsqcup_{e \in \priorTest_\tau(\cE)} \psi(e) \sqcup \bigsqcup_{e \in \priorBar_{a_1}(\cE)} \psi(e) \sqcup \ldots \sqcup  \bigsqcup_{e \in \priorBar_{a_m}(\cE)} \psi(e) \\
                      & = & \bigsqcup_{e \in \priorTest_\tau(\cE')} \psi(e)
  \end{eqnarray*} 
  \end{proof}

\smallskip
\noindent 
In addition, we need to (define and) show the soundness of the proof rule \textsc{Parallel}. This is a bit more complex. 
We start with two propositions on the existence of memories and on thread states belonging to combined and individual event structures. 

\begin{proposition} \label{prop:existence} 
   Let $\cE_1,\ldots,\cE_n$ be local assertions of threads $\tid_1,\ldots,\tid_n$, respectively. If $\cE_1,\ldots,\cE_n$ are synchronizable and $C$ is an interference free configuration of $\conf(\cE_1 ||\ldots|| \cE_n)$, then there exists some memory $M$ and thread states $ts_1,\ldots,ts_n$ (of threads $\tid_1,\ldots,\tid_n$, respectively) such that $\langle ts_i, M \rangle \in \sem{\cE_C}$, $i=1,\ldots,n$.  
\end{proposition}

\begin{proof}
 Let $C$ be an interference free configuration of $\cE_1 ||\ldots|| \cE_n$. Let $\prec$ be the total memory consistent linearizable order. We assume it to be $e_\inikw \prec e_1 \prec \ldots \prec e_m$ for $E_C = \{ e_\inikw, e_1, \ldots, e_m\}$. We now construct the following memory $M$:
 \begin{align*}
     M(0) & = \inikw \\
     M(i) &= \langle x,\val\rangle_\tau \quad \text{ if } \lfun(e_i) \in \{\ff{\tau}{x}{\val},\prm{\tau}{x}{\val}\}
 \end{align*}
 With this, we let $\psi$ map $e_\inikw$ to 0 and $e_i$ to $i$. By construction and $C$ being interference free, $\psi$ is (i) initializing at zero, (ii) consecutive, (iii) content preserving, (iv) flow preserving, and (v) memory-constraints preserving. The thread states $ts_1,\ldots,ts_n$ can then simply be calculated from $M$, $\psi$ and $C$. 
\end{proof}

The second proposition is a sort of compositionality result of parallel composition with respect to the global state (views and memory) of a program and its projection onto the threads and their individual proof outlines.  

\begin{proposition} \label{prop:localize} 
  If $\cE_1$ to $\cE_n$ are synchronizable and $C$ is an interference free configuration of $\conf(\cE_1 || \ldots || \cE_n)$, then $\langle ts_i, M \rangle \in \sem{\cE_C}$ implies $\langle ts_i, M \rangle \in \sem{\cE_i}$, $i\in \{1, \ldots, n\}$. 
\end{proposition}

\begin{proof} 
   Let $\langle ts_i, M \rangle \in \sem{\cE_C}$. 
   Then there exists a mapping $\psi : \Pr(\cE_C) \cup \Ff(\cE_C) \cup \{ e_\inikw \} \rightarrow dom(M)$. 
   Out of $\psi$, we construct $\psi_i : \Pr(\cE_i) \cup \Ff(\cE_i) \cup \{ e_\inikw \} \rightarrow dom(M)$ by letting 
   \begin{align*}
       \psi_i:  \quad & e_\inikw   \mapsto 0   \\
                & e_i   \mapsto t &\quad \text{ if } \exists (e_1,\ldots, e_i, \ldots, e_n) \in E_C: \psi(e_1,\ldots, e_i, \ldots, e_n) = t 
   \end{align*}
   We have to show that this $\psi_i$ is well defined and total. 
   \begin{itemize}
       \item it is total, because $C$ is thread-covering,
       \item well defined:\\
         assume exists another, different event
         $(e_1',\ldots, e_i, \ldots, e_n') \in E_C$. Since
         $e_i \neq *$, we have (by the definition of parallel
         composition) that
         $$(e_1,\ldots, e_i, \ldots, e_n) \# (e_1', \ldots, e_i,
         \ldots,e_n')$$ and then by definition of configuration we
         cannot have both events in $C$.
   \end{itemize}
   We further need to show that $\psi_i$ satisfies the conditions (i) to (v)  of $\sem{\cE_i}$. 
   \begin{enumerate}
       \item it initializes at zero: by construction, 
       \item consecutive for $\tid_i$: follows because $\psi$ is consecutive, 
       \item content preserving: by construction,
       \item flow order preserving:
          let $e_i \causal^+_i d_i$. Because $C$ is thread-covering, there exist events $(e_1,\ldots,e_i, \ldots,e_n),$ $(d_1,\ldots, d_i, \ldots, d_n) \in E_C$ such that 
          \[(e_1,\ldots,e_i, \ldots,e_n) \causal^+ (d_1,\ldots, d_i, \ldots, d_n).\] 
          Hence $\psi(e_1,\ldots,e_i, \ldots,e_n) < \psi(d_1,\ldots, d_i, \ldots, d_n)$ which implies $\psi_i(e_i)< \psi_i(d_i)$ by construction. 
        \item memory constraints preserving: let $d_i \rdflow{L}_i e_i$, then there exists
          $$(d_1,\ldots, d_i, \ldots, d_n) \rdflow{L'} (e_1,\ldots,e_i, \ldots,e_n)$$ for some $L' \supseteq L$. Hence for all $t$ such that $\psi(d_1,\ldots, d_i, \ldots, d_n) < t < \psi(e_1,\ldots,e_i, \ldots,e_n)$ we get $M(t).loc \neq (d_1,\ldots, d_i, \ldots, d_n).loc$. Hence we also have $\forall t: \psi_i(d_i) < t < \psi_i(e_i)$: $M(t).loc \neq d_i.loc$.
   \end{enumerate}
   We further need to show that the promises are ok: \\
   We have $ts_i.prom = M_{\tid_i} \setminus \psi(\Ff_{\tid_i}(\cE_C))$. By the definition of labelling in parallel composition and the fact that all $\cE_j$, $j \neq i$, contain no events labelled $\ff{\tid_i}{\cdot}{\cdot}$\footnote{This holds because the proof rules only introduce fulfill events of a thread $\tid$ in the proof outline of $\tid$.}, we get that for all $(e_1,\ldots,e_i, \ldots,e_n) \in E_C$ s.t.~$\lfun(e_1,\ldots,e_i, \ldots,e_n) = \ff{\tid_1}{\cdot}{\cdot}$ we have $\lfun_i(e_i) = \ff{\tid_i}{\cdot}{\cdot}$. Furthermore $\psi(e_1,\ldots,e_i, \ldots,e_n) = \psi_i(e_i)$. Hence $M_{\tid_i}\setminus \psi(\Ff_{\tid_i}(\cE_C)) = M_{\tid_i} \setminus \psi_i(\Ff_{\tid_i}(\cE_i))$. \\
   Next to the views: \\
   First $v_C$. We have
   $$
   ts_i.v_C = \bigsqcup_{(e_1,\ldots,e_n) \in \priorTest_{\tid_i}(\cE_C)} \psi(e_1,\ldots,e_n).$$ 
   Recall that  $\priorTest_{\tid_i}(\cE_C) = \{ (e_1,\ldots,e_n) \in \Pr(\cE_C) \cup \Ff(\cE_C) \cup \{e_\inikw\} \mid \exists (e_1',\ldots,e_n') \in last_{\test{\tid_i}{\cdot}}(\cE_C): (e_1,\ldots,e_n) \causal^+ (e_1',\ldots,e_n')\}$. 
   First, note that $e_j' = *$, $j \neq i$ (since this is a test-labelled event which is not synchronized). 
   Second, note that we do not necessarily have $e_i \causal_i^+ e_i'$. However, there exists a sequence  $(e_1^1,\ldots,e_n^1), \ldots, (e_1^m,\ldots,e_n^m)$ such that 
   \[ (e_1,\ldots,e_n) = (e_1^1,\ldots,e_n^1) \causal \ldots \causal (e_1^n,\ldots,e_n^m) = (e_1',\ldots,e_n') \] 
   with some $j, 1 \leq j \leq m-1$ s.t.~$(e_1^j,\ldots, e_n^j) \in sync(\cE_1,\ldots,\cE_n)$ and 
   \[ e_i^j  \causal_1 \ldots \causal_1 e_i^n \ . \]
   We furthermore have $\psi(e_1^j,\ldots,e_n^j) = \psi_i(e_i^j)$ and $\psi(e_1,\ldots,e_n) < \psi(e_1^j,\ldots,e_n^j)$ (by $\psi$ being flow preserving). 
   Thus $\bigsqcup_{e \in \priorTest_{\tid_i}(\cE_C)} \psi(e) = \bigsqcup_{e \in \priorTest_{\tid_i}(\cE_i)} \psi_i(e)$. \\
   The other views have a similar reasoning, because $\priorBar$ and $\priorFnc$ are also sets of events prior to an event of the form $(*, \ldots, *,\cdot,*, \ldots, *)$. Moreover, if $(e_1,\ldots,e_n) \in \Ff(\cE_C)$, then $e_i \in \Ff(\cE_i)$. \\
   Finally, the register values. We need to show that $\sem{a}_{\cE_C} = \sem{a}_{\cE_i}$ for all registers $a \in R(\tid_i)$. Recall that  $\sem{a}_{\cE_C}$ is calculated via $\priorBar_a$. As $a$ is local to $\tid_i$, an ordering with respect to a $\barrx{a}{\cdot}$-labelled event has to come from the event structure $\cE_i$. Hence, the same value is calculated in $\cE_C$ and $\cE_i$. 
\end{proof}

\noindent These two propositions next help us to establish our first main result which is the soundness of the rule of parallel composition. 

\soundness*
\begin{proof}  
Note first that by Proposition~\ref{prop:existence} such states $\langle ts_i, M\rangle$ always exist for interference free configurations. 

For the proof, we basically show the existence of an execution of this form: 
\[ \langle \vec{T}_0, M_0 \rangle \underbrace{\longrightarrow^*}_{\text{all promises}} \langle \vec{T}_0^p, M \rangle \underbrace{\longrightarrow^*}_{\text{the statements}} \langle \vec{T}, M \rangle \] 
First note, that \cref{lemma:promises-first} tells us that such a trace which first executes all promises and then the remaining operations exists. We first show  $\langle \vec{T}_0,M_0 \rangle \trans {}^* \langle \vec{T}_0^p,M \rangle$. 
For this, we let $\vec{T}_0^p$ be the thread pool coinciding with $\vec{T}_0$ except for the promise sets which are $ts_i.prom = M_{\tid_i}$, $i \in \{1, \ldots, n\}$. Assume $\#M = m$, $M(j) = \langle x:=\val\rangle_\tid$ and let $op_j=prm_\tid(x,\val)$, $1 \leq j \leq m$.  Then 
\[ \langle \vec{T}_0,M_0 \rangle \trans{op_1 \ldots op_m} \langle \vec{T}_0^p,M \rangle \] 
(promises are made in the order as given by $M$) and all these steps are certified as the promises will finally be empty. 

Next, we consider the steps of statements. For this, first let $ts_i^p$ be $\vec{T}_0^p(\tid_i).tstate$ be the state of thread $\tid_i$ after having made the promises. Let $op_1^i \ldots op_{k_i}^i$, $i\in \{1, \ldots,n\}$, and $op_j^i \neq \prm{\tid_i}{\cdot}{\cdot}$, be the sequence of non-promise operations such that $S_i \trans {op_1^i \ldots op_{k_i}^i}_{\tid_i} \skipst$, $i\in \{1,\ldots, n\}$.  
From  $\langle ts_i,M \rangle \in \sem{\cE_i}$ and Prop.~\ref{prop:localize}, we know that $\langle ts_i,M \rangle \in \sem{\cE_i}$. 
By soundness of local proof rules, we get the existence of thread states $ts_i^0, \ldots, ts_i^{k_i}$ for all threads such that 
\[ \langle ts_i^p, M \rangle = \langle ts_i^0,M \rangle \trans {op_1^i}_{\tid_i} \ldots \trans {op_{k_i}^i}_{\tid_i} \langle ts_i^{k_i}, M \rangle = \langle ts_i,M \rangle \] 
These steps now have to be lifted to program steps of the parallel program.  We construct a sequence of thread pools $\vec{T}_1^0, \ldots, \vec{T}_1^{k_1}$, $\ldots$, $\vec{T}_n^1, \ldots, \vec{T}_n^{k_n}$ sequentially executing all thread programs by letting (1) $\vec{T}_\ell^j(\tid_i)=(S_i,ts_i^p)$ for $\ell < i$, (2) $\vec{T}_\ell^j(\tid_i)=(\skipst,ts_i)$ for $\ell > i$, and (3) $\vec{T}_\ell^j(\tid_i)=(S_i',ts_i^j)$ for $\ell = i$ for some $S_i'$ such that $S_i \trans {op_i^1 \ldots op_j^i}_{\tid_i} S_i'$. 
Note that $\vec{T}_0^p = \vec{T}_1^0$ and $\vec{T} = \vec{T_n}^{k_n}$. As $ts_i.prom = \emptyset$, all program steps are certified and we get 
\[ \vec{T}_1^0 \trans {op_1^1 \ldots op_{k_1}^1}_{\tid_1} \vec{T}_1^{k_1} \trans {op_1^2 \ldots op_{k_2}^2}_{\tid_2} \ldots  \trans {op_1^n \ldots op_{k_n}^n}_{\tid_n} \vec{T}_n^{k_n} \]
Hence also $\langle \vec{T}_0^p,M \rangle \trans {}^* \langle \vec{T},M \rangle$. 
\end{proof}

\subsection{Completeness}

Finally, we prove completeness of our proof calculus. We must prove that for every trace of a program we can construct local event structures, whereby the parallel composition of the local event structures contains an interference-free configuration. We construct these local event structures by induction following the execution trace. 

\completeness*

For the proof, we need a number of well-formedness conditions on the event structures our proof rules construct. First of all, we slightly extend the action labels of events: instead of $\prm{\tau}{x}{\val}$ we use $\prmts{\tau}{t}{x}{\val}$ to denote the timestamp $t$ of the read value in memory,  and we similarly use $\ffts{\tau}{t}{x}{\val}$. 
In a trace $tr=\langle \vec{T}_0, M_0 \rangle \trans {op_1} \ldots \trans {op_2} \langle \vec{T}, M \rangle$ we define $\mathbb{T}_{rd}(tr,\tau)$ to be the set of timestamps that thread $\tau$ {\em reads} from,  
similarly  $\mathbb{T}_{\mathit{ff}}(tr,\tau)$ for the timestamps of the fulfills. 
With these extensions at hand, we say that an event structure $\cE =(E,\causal,\#,\lres,\lfun)$ of thread $\tau$ is {\em timestamp closed} if the following holds:   $\forall t,t' \in \mathbb{T}_{rd}(tr,\tau), e \in E$, if   $\lfun(e) = \prmts{\cdot}{t}{\cdot}{\cdot}$ and $t' < t$, $t' \neq 0$, then there exists an $e' \in E$ with $\lfun(e') \in \{\prmts{\cdot}{t'}{\cdot}{\cdot}, \ffts{\tau}{t'}{\cdot}{\cdot}\}$;  
$\cE$ is {\em timestamp ordered} if $\forall t,t' \in \mathbb{T}_{rd}(tr,\tau) \cup \mathbb{T}_{\mathit{ff}}(tr,\tau)$ we have: if $t'<t$, then $t'$-labelled events cannot be flow-after $t$ labelled events;
$\cE$ contains a {\em modification order} for all locations $x$ if for all $x \in \Loc$, $\causal^+|_{\Act^x \times \Act^x}$ is a total order.  
Finally, we need some conditions on the read events occuring in the event structures. Recall that the rule \textsc{PR-ReadNew} may introduce a number of read events (ordered in a chain) of which only the last in the chain is used for reading at that time. This reading is marked by attaching a bar event flow after the read. We define a read event $e$ in $\cE$ to be {\em unbarred}, $ubr_\cE(e)$, if $\lfun(e) = \prm{\cdot}{x}{\cdot} \wedge \neg \exists e'$ s.t.~$\lfun(e') = \barrx{\cdot}{x} \wedge e \causal_\cE e'$. We say that all {\em read chains end in a bar} in $\cE$ if the following holds: 
$\forall e \in E: ubr_\cE(e) \imp \exists m, \exists e_1, \ldots, e_m \in \Pr(\cE), e_1 = e, e_i \causal_\cE e_{i+1} \wedge \neg ubr(e_m)$. 

\begin{proof}
Assume $tr: \langle \vec{T_0}, M_0 \rangle \trans {op_1} \langle \vec{T_1}, M_1 \rangle \ldots \trans {op_m}  \langle \vec{T_m}, M_m \rangle = \langle \vec{T},M \rangle$ is the trace such that all promises are done at the beginning, and let $op_j$ be the last such promise operation leading to $\langle \vec{T}_j, M_j \rangle$. Note first that $M_j = M_{j+1} = \ldots = M_m = M$ (only promises change the memory). We now inductively construct the event structures $\cE_k^i$ of every thread $k$, $1 \leq k \leq n$ (and by this the proof outlines) for all indizes $i=j$ to $m$. 
In every step, the event structures $\cE_k^i$ are (i) timestamp closed, (ii) timestamp ordered,   (iii) contain a modification order,  (iv) all read chains end in a bar and moreover (v) $\langle \vec{T}_i(\tau_k), M \rangle \in \sem{\cE_k^i}$ using the mapping $\psi: \Ff(\cE_k^i) \cup \Pr(\cE_k^i) \cup \{e_\inikw\} \rightarrow dom(M)$ to be 
\begin{align*}
   e_\inikw & \mapsto 0 \\
   \prmts{\cdot}{t}{\cdot}{\cdot} & \mapsto t \\
   \ffts{\cdot}{t}{\cdot}{\cdot} & \mapsto t 
\end{align*}

\begin{description}
  \item[Induction Start] $i=j$: $\cE_k^i = \Ini$. As the event structure $\Ini$ contains just a single event, it is timestamp closed, timestamp ordered and contains a modification order. Moreover, $\langle \vec{T}_j(\tau_1), M \rangle \in \sem{\Ini}$ because the event $e_\inikw$ gets mapped to 0 by $\psi$, and hence all views are 0, and the promise set contains all entries in $M$ of $\tau_1$. 
  \item[Induction Step] Assume $\cE_k^i$ with properties (i) to (iv) has been constructed. We now have different cases depending on the next operation $op_{i+1}$. Let $\vec{T}_i(\tau_k) = (S^i,ts^i)$. 
  \begin{enumerate}
      \item $op_{i+1}$ is not an operation of $\tau_k$:  \\
          Then $\cE_k^{i+1} = \cE_k^i$. As the event structure is not changing and the state of thread $\tau_k$ is not changing, all properties are trivially preserved. 
      \item $op_{i+1} = \kwc{dmb}$ is an operation of thread $\tau_k$:  \\
         Then $\cE_k^{i+1} = \cE_k^i \oplus \fence_{\tau_k}$.  
         As the new event is neither a read nor a fulfill (which would have timestamps), event structure $\cE_k^{i+1}$ is still timestamp closed, timestamp ordered, contains a modification order and all read chains end in bar. As for the views: there are two views affected by the fence operation, $v_{read}$ and $v_{wNew}$. We only consider these here. First,  $v_{read}$. By the operational semantics, we know that $ts^{i+1}.v_{read} = ts^i.v_{read} \sqcup ts^i.v_{wOld}$. We need to show that this fits to the semantics of $\cE_k^{i+1}$ given the semantics of $\cE_k^i$ and operation $\oplus$.
         \begin{small}
         \begin{eqnarray*} 
            ts^i.v_{read} \sqcup ts^i.v_{wOld} & = & \bigsqcup_{e \in
            \big(\priorFnc_{\tau_k}(\cE_k^i) \cap \Ff_\tid(\cE_k^i)\big) \cup \priorBar_{\tau_k}(\cE_k^i)} \psi(e) \sqcup \bigsqcup_{e \in \Ff_\tau(\cE_k^i)} \psi(e)  \\
                   & = & \bigsqcup_{e \in \Act_{\tau_k}(\cE_k^i)} \psi(e) \\
                   & = & \bigsqcup_{e \in \Act_{\tau_k}(\cE_k^i) \cup \priorBar_{\tau_k}(\cE_k^i)} \psi(e) \\
                   & = & \bigsqcup_{e \in
            \big(\priorFnc_{\tau_k}(\cE_k^{i+1}) \cap \Ff_\tid(\cE_k^{i+1})\big) \cup \priorBar_{\tau_k}(\cE_k^{i+1})} \psi(e)
         \end{eqnarray*}
       \end{small}

         The last term corresponds to $ts^{i+1}.v_{read}$ when calculated from $\cE_1^{i+1}$. For $v_{wNew}$ we get almost the same reasoning as the new fence event is placed flow after all but test events.  
      \item $op_{i+1} = \ff{\tau_k}{x}{\val}$ is an operation of thread $\tau_k$:    \\
        By the operational semantics, this transition results from a statement $\store{x}{a}$ for some register $a$ such that $ts^i.regs(r) = \val@v_a$. Let $t$ be the timestamp of the promise in $M$ being fulfilled, i.e. $t \in ts^i.prom$.  
        Since $\langle \vec{T}_i(\tau_k), M_i \rangle \in \sem{\cE_k^i}$ (by induction hypothesis), we get $\sem{a}_{\cE_k^i} = \val$. Hence proof rule \textsc{Write} is applicable and we let $\cE_k^{i+1} = \cE_k^i \oplus^a \ffts{\tau_k}{t}{x}{\val}$. We show properties (i) to (v). By induction hypothesis, $\cE_k^{i+1}$ is timestamp closed as the additional event is not a read event. It is timestamp ordered as the new event is added flow after all events $e'$ with $\lfun(e') \in \Act^x$, and we know by the operational semantics that $ts^i.coh(x) < t$, and by $\sem{\cE_k^i}$ that $ts^i.coh(x)$ is the maximum of all fulfills and reads to $x$. It moreover maintains the existence of a modification order for $x$ due to the same reasons. All read chains end in bar because no new reads are added. \\
        Next to the state after the fulfill: the operational semantics changes the view $coh(x)$ and $v_{wOld}$ as well as $prom$ and keeps the rest. We need to show that this after state is in $\sem{\cE_k^{i+1}}$. First, $coh(x) = t$ (by operational semantics) and this would be derived from $\cE_k^{i+1}$ as the new event $\ffts{\tau_k}{t}{x}{\val}$ is maximal on all $x$-labelled events. 
        Also, since this event is now in the event structure, $t$ is not in $prom$ anymore. Finally, $v_{wOld}$ is derived from the event structure as the maximum of all timestamps of fulfills. This includes $t$ and hence $v_{wOld} := v_{wOld} \sqcup t$ holds. For the remaining views: $v_C$ is unchanged because no new test events are added; $v_{wNew}$ is unchanged because no fence is added; all $v_a$'s are kept as no bar events are added and $v_{read}$ is unchanged as neither fences nor bars are added.
    \item $op_{i+1}=lst(a,exp)$ is an operation of thread $\tau_k$:  \\ 
        Then $\cE_k^{i+1} = \cE_k^{i} \oplus \barrx{a}{exp}$. As the new event is neither a read nor a fulfill, $\cE_k^{i+1}$ is still timestamp closed, timestamp ordered, contains a modification order and all read chains end in bar. Further $\sem{exp}_\cE=\sem{exp}_{ts^{i+1}.regs} @ v$ for some view $v$ and $ts^{i+1}.v_a = \bigsqcup_{e\in \priorBar_a(\cE_k^{i+1})} \psi(e)$. Both equations follow directly from the construction in the proof of Proposition~\ref{prop:register}.  
    \item $op_{i+1}=asm(\bexp)$ is an operation of thread $\tau_k$: \\  
        Since $\langle \vec{T}_i(\tau_k), M_i \rangle \in \sem{\cE_k^i}$ and $\sem{\bexp}_{ts^i.regs} = \true @ v$, for some view $v$, we get $\sem{\bexp}_{\cE_k^i} = \true$. Hence proof rule \textsc{Assume} is applicable and we let $\cE_k^{i+1} = \cE_k^i \oplus \test{\tau_k}{\bexp}$. As the new event is neither a read nor a fulfill, $\cE_k^{i+1}$ is still timestamp closed, timestamp ordered, contains a modification order and all read chains end in bar. Now we only need to show that $ts^{i+1}.v_C=\bigsqcup_{e\in \priorTest_{\tau_k}(\cE_k^{i+1})} \psi(e)$ which follows directly from the construction in the proof of Proposition~\ref{prop:assume}. 

    \item $op_{i+1} = rd(x,\val)$ is a transition of thread $\tau_k$ executing instruction $\load{a}{x}$ thereby reading from timestamp $t$: \\ 
    Let $T=\{ t_1, \ldots, t_q \}$ be the set of timestamps that thread $\tau_k$ will read in $tr$ {\em after and including} $op_{i+1}$
    for which there is furthermore no timestamp in $\mathbb{T}_{\mathit{ff}}(tr,\tau_k)$.
    Let $T_{\leq t} = \{ t' \in T \mid t' \leq t \}$ and let $t^1 < t^2 < \ldots < t^p$ be an ordering of that set. Note that $t^p = t$. 
      By timestamp closedness, there is some index $\ell$, $0 \leq \ell \leq p$, such that all timestamps $t' \in \mathbb{T}_{rd}(tr,\tau_k)$ with $t' \leq t^\ell$ have read or fulfill events in $\cE_k^i$. 
       Two cases need to be considered: \\
      If $\ell=p$, i.e.~read or fulfill events for all $T_{\leq t}$ are in $\cE_k^i$, we apply rule \textsc{PR-ReadEx}.  
       We let $e' = last_{\Act^x}(\cE_k^i)$. By existence of a modification order, $e'$ is uniquely defined. We need to show that the timestamp of $e'$ is $t$. Assume the contrary, i.e.~$e'$ has a timestamp $t' \neq t$. Let $e$ be the event in $\cE_k^i$ with timestamp $t$. $e$ is also an event on location $x$, i.e.~by modification order and $e' \in last_{\Act^x}(\cE_k^i)$, we have to have $e \causal^+ e'$. By timestamp ordering, we thus have $t < t'$. By the operational semantics, we get $t' > ts^i.v_{read} \sqcup ts^i.coh(x)$. As $\langle ts^i,M \rangle \in \sem{\cE_k^i}$, $e'$ has to be an unbarred read (because in case of a fulfill, $ts^i.coh(x) \geq t'$; in case of a barred read $ts^i.v_{read} \geq t'$). 
       However, in $\cE_k^i$ all read chains end in a bar, hence $ts^i.v_{read}$ is greater or equal the timestamp of the last event in the chain, and hence greater or equal $t'$. Contradiction. Hence the timestamp of $last_{\Act^x}(\cE_k^i)$ is $t$.

       We now let $e = last_{\Act^x}(\cE_k^i)$ to match the notation of the rule. 
      As $\langle ts^i,M \rangle \in \sem{\cE_k^i}$ using the specific $\psi$ and $M(t)$ being $\langle x:=\val \rangle _\tau$, $e$ is labelled 
       $\prmts{\tau}{t}{x}{\val}$, $\ffts{\tau}{t}{x}{\val}$ or $\inikw$ (with $\val=0$). We then let $\cE_k^{i+1} = rstr_e^x(\cE_k^i \oplus \barrx{a}{x})$.  
       We need to show this new event structure to satisfy all properties. First, as we neither add new read nor new fulfill events, $\cE_k^{i+1}$ is still timestamp closed, timestamp ordered and contains a modification order. As we only add additional bar events, all read chains continue to end in a bar. 
       We need to show $\langle ts^{i+1}, M \rangle \in \sem{\cE_k^{i+1}}$. The reasoning for promises and views follows that used in the proof of Proposition~\ref{prop:readex}. Finally, we need to take a look at the new location restriction in $\cE_k^{i+1}$. By the operational semantics, we get 
       \[ \forall t', t < t' \leq ts^i:v_{read} \sqcup ts^i.coh(x) \imp M(t') \neq x \]
       By $\langle ts^i, M \rangle \in \sem{\cE_k^i}$, we thus get 
       \[ \forall t', t < t' \leq \bigsqcup_{f \in (\priorFnc_\tau(\cE^i_k) \cap \Ff_\tau(\cE^i_k)) \cup \priorBar_\tau(\cE^i_k) \cup \Ff_\tau^x(\cE_K^i)} \psi(f) \quad \imp M(t') \neq x\] 
       Hence $M$ satisfies the additional memory constraint placed by $restr_e^x(\cE_k^i \oplus \barrx{a}{x})$. 
       
       Second case. If $\ell < p$, we apply rule \textsc{PR-ReadNew} and let $\cE'$ be the event structure with $E' = \{e^{\ell+1}, \ldots, e^p\}$, $ e^{\ell + 1} \causal' \ldots \causal' e^p$, $\lfun'(e^m) = \prmts{\tau}{t^m}{x}{u}$ such that $M(t^m) = \langle x:=u\rangle_\tau$. By construction, $\cE'$ is sequential, $\lfun'(E') \subseteq \Act^{Rd}$, $last_{\Act^x}(\cE') = e$ is an event labelled $\prmts{\tau}{t}{x}{\val}$ and hence rule \textsc{PR-ReadNew} is applicable. We now let $\cE_k^{i+1} = (\cE_k^i \oplus \cE') \oplus \barrx{a}{x}$. 
       Next to the properties. In $\cE_k^{i+1}$ all read chains end in a bar (as the additional read chain ends in $e$ followed by $\barrx{a}{x}$; $\cE_k^{i+1}$ contains a modification order (as $\cE'$ is sequential and hence contains a modification order, $\cE_k^i$ contains a modification order and $\cE_k^i \oplus \cE'$ orders all events on same locations. $\cE_k^{i+1}$ is timestamp ordered: $\cE'$ is timestamp ordered by construction, $\cE_k^i$ is timestamp ordered by induction, now assume that here is some event $f \in \cE_k^i$ and $g \in \cE'$ with timestamps $t^f$ and $t^g$ such that $f \causal g$ but $t^g > t^f$. By $\oplus $ on event structures, there is some location $y$ such that $f$ is a fulfill on $y$ and $g$ a read of $y$. As $\langle ts^i,M \rangle \in \sem{\cE_k^i}$, $ts^i.coh(y) \geq t^f$. However, then the trace $tr$ cannot later have a read of $y$ from a timestamp earlier than $t^f$ (by view monotonicity and operational semantics). Hence $\cE_k^{i+1}$ is timestamp ordered. Timestamp closedness follows by construction of $\cE'$ which assumes there are already events in $\cE_k^i$ for all timestamps $t' \leq t^\ell$.  
       
       It remains to be shown that $\langle ts^{i+1}, M \rangle \in \sem{\cE_k^{i+1}}$. 
      Note first that $t$ is larger than all timestamps of all events in $\cE_k^i$. As $\langle ts^i, M \rangle \in \sem{\cE_k^i}$, we thus have 
      \[ ts^i.v_{read} < t \wedge ts^i.v_a < t \wedge \forall y \in \Loc: ts^i.coh(y) < t \]
      Hence $ts^{i+1}.v_{read} = ts^{i+1}.v_a = ts^{i+1}.coh(x) = t$. This is consistent with $\cE_k^{i+1}$ because $\cE_k^{i+1}$ has a bar event directly flow after the maximal read event $e^p$ in $\cE'$ which has timestamp $t$. Views $v_C, v_{wOld}$, $v_{wNew}$, $coh(y)$ for $y \neq x$ do not change as neither new test nor fulfill nor fence nor $\barrx{\cdot}{y}$ events are added. Finally, $ts^{i+1}.prom = ts^i.prom$ and this still matches $\cE_k^{i+1}$ as no fulfill events are added.  
  \end{enumerate}

  We next let $\cE_k=\cE_k^m$, $1 \leq k \leq n$. We need to show that $\cE_1$ to $\cE_n$ are synchronisable and that there exists some interference free configuration $C \in (\cE_1 || \ldots || \cE_n)$. \\
  $\cE_1$ to $\cE_n$ are synchronisable because the read events in the event structure correspond to the read operations occurring in the execution, they read from memory and every entry in memory has a fulfill operation in the execution (as all promise sets of threads are empty at the end) and hence a fulfill event in the event structures. \\
  We construct the configuration $C$, i.e.~a set of events. Note that we here mix events and their labellings. 
  \begin{enumerate}
      \item $(e_\inikw, \ldots, e_\inikw) \in E_C$, 
      \item Let $t \in dom(M) \setminus \{0\}$. 
      For every $t$ we construct a single event $e_t = (e_1,\ldots,e_n)$ as follows: \\
      (i) $e_k = *$ if $t \notin \mathbb{T}_{rd}(\tau_k) \cup \mathbb{T}_{\mathit{ff}}(\tau_k)$. \\
      (ii) $e_k = \ffts{\tau_k}{t}{x}{\val}$ if $t \in \mathbb{T}_{\mathit{ff}}(\tau_k)$. \\
      (iii) $e_k = \prmts{\tau_i}{t}{x}{\val}$ if $t \in  \mathbb{T}_{\mathit{rd}}(\tau_k)$ and $t \in \mathbb{T}_{\mathit{ff}}(\tau_i)$ else. 
      \item for every $e_k \in \cE_k$ with $\lfun_k(e_k) \in \{\fence_{\tau_k}, \test{\tau_k}{\cdot}, \barrx{a}{\cdot} \mid a \in R(\tau_k)\}$, we add one event $(*, \ldots, e_k, \ldots, *)$ to $C$.  
  \end{enumerate}
  We first show that $C \in \conf(\cE_1 || \ldots || \cE_k)$, i.e.~it is a configuration. 
  \begin{itemize}
      \item $C$ is cycle-free: \\
         By construction, all $\cE_k$ are cycle-free. So, cycles may only arise because of synchronization. Assume there are timestamps $t_1, t_2$ and events $(d_1,\ldots, d_n)$ (labelled with $t_1$) and $(e_1,\ldots, e_n)$ (labelled with $t_2$) such that $(d_1,\ldots, d_n) \causal^+ (e_1,\ldots, e_n)$ and $(e_1,\ldots, e_n) \causal^+ (d_1,\ldots, d_n)$. This ordering has to come from some $\cE_k$, hence contradicts timestamp orderedness of all $\cE_k$. 
      \item $C$ is conflict free: \\
        This holds because all $\cE_k$ are conflict-free and no event gets paired with two or more different events. 
      \item $C$ is left-closed up to conflicts. Let $(d_1,\ldots,d_n), (e_1,\ldots,e_n) \in E_{\cE_1 || \ldots || \cE_n}$ such that we have $(d_1,\ldots, d_n) \causal (e_1,\ldots,e_n)$ and $(e_1,\ldots, e_n) \in C$, however, $(d_1,\ldots, d_n) \notin C$. Assume w.l.o.g.~that the flow is coming from $d_k \causal_k e_k$. By construction of $C$, the event $(d_1,\ldots, d_n)$ has to have either (1) $\lfun_k(d_k) = \prmts{\tau_i}{t}{\cdot}{\cdot}$ and there exists some $j$ such that $d_j = *$ but $t \in \mathbb{T}_{rd}(\tau_j)$ or $d_i = *$ or (2) $\lfun_k(d_k) = \ffts{\tau_k}{t}{\cdot}{\cdot}$ and exists some $j$ such that $d_j = *$ but $t \in \mathbb{T}_{rd}(\tau_j)$ (otherwise the construction of $C$ would have included the event). In both cases $e_t \# (d_1, \ldots,d_n)$ and $e_t \causal (e_1, \ldots, e_n)$ guaranteeing left-closedness up to conflicts.

  \end{itemize}
 Next, we look at interference freedom. Note that by construction $C$ is thread-covering. It contains no unsynchronised reads as the read events in the trace need to read from a promise in memory, hence there has to exist a fulfill. 
 For the ordering $\prec$ we now take the ordering as induced by the timestamps on events. 
 By timestamp ordering we then get that $\causal_C^* \cap (\Act^x(E_C) \times \Act^x(E_C)) \subseteq \prec$. Note furthermore that by construction $\langle \vec{T}(\tau_k), M \rangle \in \sem{\cE_k}$, $k\in \{1,\ldots,n\}$, using the mapping $\psi$ as detailed above. Hence, the location restrictions are also met. 
\end{description}

\end{proof}


 \end{document}